\documentclass[aps,pre,a4paper,twocolumn,groupedaddress,showpacs,showkeys,preprintnumbers,floatfix,dvips]{revtex4}

\usepackage[utf8]{inputenc} \usepackage[T1]{fontenc}
\usepackage{psfrag} 
\usepackage{times} \usepackage{graphicx} \usepackage{color}
\usepackage{amsmath} \usepackage{amscd} \usepackage{amssymb} \usepackage{amsthm}
\usepackage{ragged2e} \usepackage{subfigure} \usepackage{caption}
\usepackage{nicefrac}
\usepackage{captcont}

\DeclareCaptionJustification{reallyjustified}{\justifying}
 \newtheorem{definition}{Def}
\newtheorem{theorem}{Theorem}
\newtheorem{corollary}{Corollary}
\newcommand{\oo}{\mathrm{o}}
\newcommand{\thickhline}{\noalign{\hrule height 1pt}}

\begin{document}

\title{Connections between the Sznajd Model with General Confidence Rules and graph theory}
\author{Andr\'e M. Timpanaro} \email[]{timpa@if.usp.br} \author{Carmen P. C.
Prado} \email[]{prado@if.usp.br} \affiliation{Instituto de F\'{i}sica,
Universidade de S\~{a}o Paulo \\ Caixa Postal 66318, 05314-970 - S\~{a}o Paulo -
S\~{a}o Paulo - Brazil} \date{\today}


\begin{abstract}
The Sznajd model is a sociophysics model, that is used to model opinion propagation and consensus formation in societies. Its main feature is that its rules favour bigger groups of agreeing people. In a previous work, we generalized the bounded confidence rule in order to model biases and prejudices in discrete opinion models. In that work, we applied this modification to the Sznajd model and presented some preliminary results. The present work extends what we did in that paper. We present results linking many of the properties of the mean-field fixed points, with only a few qualitative aspects of the confidence rule (the biases and prejudices modelled), finding an interesting connection with graph theory problems. More precisely, we link the existence of fixed points with the notion of strongly connected graphs and the stability of fixed points with the problem of finding the maximal independent sets of a graph. We present some graph theory concepts, together with examples, and comparisons between the mean-field and simulations in Barab\'asi-Albert networks, followed by the main mathematical ideas and appendices with the rigorous proofs of our claims. We also show that there is no qualitative difference in the mean-field results if we require that a group of size $q>2$, instead of a pair, of agreeing agents be formed before they attempt to convince other sites (for the mean-field, this would coincide with the $q$-voter model).
\end{abstract}
\maketitle

\section{Introduction}
\label{sec:intro}

In the last years, the interest in interdisciplinary problems has increased among physicists, creating many research areas. One of these areas is sociophysics, that studies how assumptions about the behaviour and social interactions of people in a ``microcospic level'' creates emerging social behaviours, like opinion propagation, consensus formation, properties of elections, how wealth is distributed in society, among other topics. Typical approaches include modelling using deterministic celular automata, Monte Carlo simulations of models derived from ferromagnetic models (usualy Ising and Potts), mean-field approaches and diffusion-reaction processes \cite{votante-def, sznajd-def, deffuant-def, HK-def, ochrombel-def, axelrod-def, Galam-2004, Castellano-2009}.

The Sznajd model is an opinion propagation model, originally inspired by the Ising model in a linear chain, and is typically used to model consensus formation. It has spawned many variations, including the addition of noise, contrarian-like agents and undecided voters; as well as generalizations to more than 2 states (opinions) and to arbitrary networks \cite{sznajd-def, Bernardes-2002}. In all these variations, the most defining aspect of the Sznajd model is that it gives a greater convincing power to bigger groups of agreeing agents. Even though the importance of this effect has been known by psychologists since the 1950s \cite{Asch-1951}, it is often overlooked in other opinion propagation models, for the sake of simplicity (this happens for example in the voter and in the Deffuant models \cite{votante-def, deffuant-def}).

In a recent work, we took the bounded confidence rule (that roughly says that people are only allowed to change opinions in a \emph{smooth} way) that is common to many opinion propagation models \cite{deffuant-def, HK-def, Stauffer-2001}, including the Sznajd model, and we generalized it to model biases and prejudices in discrete opinion models (these generalized rules will be called by the umbrella term \emph{confidence rules}). We applied this generalization to the Sznajd model and studied mainly the case with 3 opinions. In that work, we found a good qualitative (and in some cases quantitative) agreement between the model simulated in Barab\'asi-Albert (BA) networks \cite{rede-BA-def} and the mean-field equations, but some of the results about the mean-field were still rather sketchy. In the present paper we give rigorous proofs (the main mathematical ideas are in a sepparate section and the detailed proofs can be found in the appendices) about the structure, existence and stability of the fixed points for the mean-field version of the Sznajd model with general confidence rules, finding a connection between these properties and graph theory problems using a graph derived only from qualitative properties of the confidence rule.

The results have some counterintuitive aspects and as such we provide both numerical solutions for the mean-field equations and Monte Carlo simulations for the model in a BA network. In our calculations for the mean-field, we use a variant of the model, where at each timestep we choose $q$ agents and if they agree, they attempt to convince $r$ other agents; and show that there is no qualitative difference between all the cases with $q\geq 2$ (which includes the usual definition of the Sznajd model in an arbitrary network, with an arbitrary number of opinions, as we have considered in \cite{Timpanaro-2009}).

\section{The Sznajd model with confidence rules}
\label{sec:model}

The Sznajd model is an agent based sociophysics model for opinion propagation. In this model, a society is represented by a network (that is, a collection of nodes linked together by edges), where each node represents an agent (person), each edge is a social connection (friendship, marriage, acquaintances, etc.) and each node $i$ possesses an integer $\sigma_{i}$, between 1 and $M$, representing its opinion. In our generalization of the Sznajd model, as defined in \cite{Timpanaro-2009}, we introduce a set of parameters $p_{\sigma \rightarrow \sigma'}$ (that are fixed and completely independent with the state of the network), and at each time step the following update rule is used:

\begin{itemize}
\item A node $i$ is chosen at random, and then a neighbour $j$ of $i$ is chosen.
\item If they disagree ($\sigma_i \neq \sigma_j$), nothing happens.
\item If they agree, a neighbour $k$ of $j$ is chosen and is convinced of opinion $\sigma_i$ with probability $p_{\sigma_k \rightarrow \sigma_i}$.
\end{itemize}

We can interpret the first step as a conversation between two people that know each other, where they discuss some issue. If they disagree, none manages to convince the other. But, if they agree, they may set to convince another person that one of them knows and this person is convinced with a certain probability that depends only of its current point of view and of the opinion the pair is trying to impose.

In the original model the probability weights $p_{\sigma \rightarrow \sigma'}$ are not dependent on $\sigma$ and $\sigma'$. The reason why this probability should depend on both opinions is that, usualy an opinion includes prejudices about differing points of view (this is strongly related with the idea of cognitive dissonance in psychology \cite{Festinger-cog-dis, Elster-cog-dis, Lee-cog-dis}). This generalization allows for complex interactions among the opinions in an unified way and can be seen as a generalization of the \emph{bounded confidence} rules \cite{deffuant-def, HK-def}, as those rules can be recovered as special cases. Other modifications of the model can also be obtained this way:

\begin{itemize}
\item When $p_{\sigma\rightarrow\sigma'} = 1 \,\,\forall\,\,\sigma, \sigma'$ we have the usual model.
\item If $|\sigma - \sigma'| \leq \varepsilon \Rightarrow p_{\sigma\rightarrow\sigma'} = 1$ and $p_{\sigma\rightarrow\sigma'} = 0$ otherwise, we have bounded confidence with threshold $\varepsilon$.
\item Undecided agents can be modelled by a special state $\sigma$, such that $p_{\sigma'\rightarrow\sigma} = 0 \,\,\forall\,\,\sigma'$ (undecided agents can only be convinced).
\item Cyclic interactions, like rock, paper, scissors ($A$ convinces only $B$, that convinces only $C$, that convinces only $A$).
\end{itemize}

This generalized version of the model has $M(M-1)$ parameters, where $M$ is the number of opinions ($p_{\sigma\rightarrow\sigma}$ is irrelevant and can always be taken as 0). These parameters can be thought as the elements of the adjacency matrix of a directed weighted graph, that will be refered to as \emph{confidence rule} (we will also refer to the parameters $p_{\sigma\rightarrow\sigma'}$ in this way). So the confidence rule is a directed weighted graph, whose nodes are the opinions in the model (so a model with $M$ opinions would have a confidence rule with $M$ nodes) and the arcs represent the ways that opinions are allowed to interact. This graph is useful as a way of schematizing the opinion interactions and as we show in the next sections, it can be used to find the properties of the mean-field fixed points. An example of confidence rule with 4 opinions is given in figure \ref{fig:CR}.

\begin{figure}[hbt]
\centering
\includegraphics[width=0.7\linewidth]{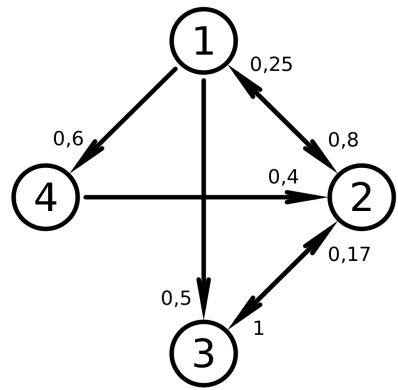}
\caption{A confidence rule for 4 opinions. Here $p_{1\rightarrow 2}=0.8$, $p_{1\rightarrow 3}=0.5$, $p_{1\rightarrow 4}=0.6$, $p_{2\rightarrow 1}=0.25$, $p_{2\rightarrow 3}=1$, $p_{3\rightarrow 2}=0.17$, $p_{4\rightarrow 2}=0.4$ and $p_{\sigma\rightarrow\sigma'}=0$ otherwise.}
\label{fig:CR}
\end{figure}

\subsection{The mean-field model}
\label{ssec:confidence}

For the analysis of the mean field case, we consider a variant of the Sznajd model, where at each timestep we choose $q$ agents at random and if they agree (meaning they are on the same state), they attempt to convince $r$ other agents (also chosen at random). If the group of $q$ agents has opinion $\sigma$, then each of the targeted $r$ agents is convinced with probability $p_{\sigma'\rightarrow\sigma}$ and retains its opinion with probability $1 - p_{\sigma'\rightarrow\sigma}$, where $\sigma'$ is the opinion the targeted agent had before the group attempted to convince it (and hence it is different for each of the $r$ agents). Adding up the probabilities of all possible processes we obtain the mean field equation in the limit of large networks:

\begin{equation}
\dot{\eta}_{\sigma} = r\sum_{\sigma'}\eta_{\sigma}\eta_{\sigma'}(\eta_{\sigma}^{q-1}p_{\sigma'\rightarrow\sigma} - \eta_{\sigma'}^{q-1}p_{\sigma\rightarrow\sigma'}),
\label{eq:mean-field}
\end{equation}
where $\eta_{\sigma}$ is the proportion of sites with opinion $\sigma$ (the deduction of this equation from the underlying Markov chain implies that $\eta$ is actually the expected value of the proportion) and a time unit corresponds to a Montecarlo timestep (MCT), that is, a number of timesteps equal to the number of sites in the network. The phase space of this flow is an $(M-1)$-simplex denoted as $Sim_M$ (that is embedded in an $M$ dimensional vector space, in order to make the equations more symmetrical), where the vertices correspond to consensus states and the other states are convex combinations of the vertices, with coeficients $\eta_{\sigma}$: 

\begin{equation}
P=\sum_{\sigma} P_{\sigma}\eta_{\sigma},
\label{eq:embedding}
\end{equation}
where $P_{\sigma}$ is the coordinate of the vertex corresponding to consensus of opinion $\sigma$ and $P$ is the coordinate in phase space of the point representing the state $(\eta_1, \ldots, \eta_M)$ (in other words, we're using a barycentric coordinate system).

The results for the mean-field fixed points can be expressed as problems regarding the existence of groups of nodes satisfying certain conditions in the confidence rule and these results are the same for all $q \geq 2$. We will give here these results for a better understanding of the simulations in section \ref{sec:simulation}, leaving the mathematical details for later. Because of the connection of these results with graph theory, a small glossary (with examples) will be useful. For the same reason, we will interchange freely the notion of a set of opinions with the notion of a set of nodes in some graph (like the confidence rule).

\subsection{Graph theory concepts and glossary}

If $G$ is a weighted graph, with adjacency matrix $G_{i\rightarrow j}$ (that is, the matrix containing the weights of the graph) and $G_{i\rightarrow j} \geq 0$, one can define its directed skeleton $\mathrm{Sk}_{dir}(G)$ as the directed graph with adjacency matrix:

\[
S_{i\rightarrow j} = \left\{\begin{array}{l}
0,\mbox{  if  }G_{i\rightarrow j} = 0\\
1,\mbox{  if  }G_{i\rightarrow j} \neq 0
\end{array}
\right..
\]

\noindent An example of skeleton is given in figure \ref{fig:skeletons}.

\begin{figure}[hbt]
\centering
\includegraphics[width=0.5\linewidth]{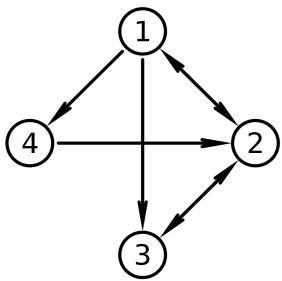}
\caption{The skeleton for the confidence rule in figure \ref{fig:CR}.}
\label{fig:skeletons}
\end{figure}

Let now $\Delta$ be a set of nodes in a directed graph $\mathcal{S}$ (typically in our problems, $\Delta$ will be a set of opinions and $\mathcal{S}$ will be the skeleton of the confidence rule). We define the following terms (we will use in the examples $\mathcal{S}$ equal to the graph in figure \ref{fig:skeletons}):

\begin{itemize}
\item The predecessor set of $\Delta$, denoted by $\Delta_{-}$, is the set of nodes in $\mathcal{S}$ that point to some node in $\Delta$. For example, $\{3\}_{-} = \{1,4\}_{-} = \{1,2\}$ and $\{2\}_{-} = \{1,3,4\}$.
\item Analogously, the successor of $\Delta$, denoted by $\Delta_{+}$, is the set of nodes in $\mathcal{S}$ that are pointed by nodes in $\Delta$. For example, $\{2,3\}_{+} = \{1,2,3\}$.
\item The complement of $\Delta$, $\overline{\Delta}$ is the set of nodes in $\mathcal{S}$ that are not in $\Delta$. For example, $\overline{\{2,3\}} = \{1,4\}$.
\item $\Delta$ is an independent set iff $\mathcal{S}$ has no connections among nodes in $\Delta$. $\{1\}$ and $\{3,4\}$ are independent sets. Note that if $\Delta$ is independent, it follows that $\Delta_{-}, \Delta_{+} \subseteq \overline{\Delta}$.
\item An independent set $\Delta$ is maximal if it contains all the nodes in the graph or if the addition of any node not in $\Delta$ destroys independence. $\{3,4\}$ is a maximal independent set, while $\{3\}$ is independent but not maximal.
\item If $\Delta$ is a set of nodes from $\mathcal{S}$, then the graph induced by $\Delta$, $\mathcal{S}_{\Delta}$ is the graph whose set of nodes is $\Delta$ and whose connections are the connections between the elements of $\Delta$ that existed in $\mathcal{S}$. The graph $\mathcal{S}_{\{2,3,4\}}$ can be found in figure \ref{fig:induced}.
\item The union of 2 graphs $G$ and $H$, denoted $G\cup H$ is the graph with all the nodes of $G$ and $H$, but only connections that already existed between $G$ and $H$ (in short it means referring to 2 unrelated graphs as parts of the same graph, without changing anything else). The graph $\mathcal{S}_{\{1,2,3\}}\cup\mathcal{S}_{\{4\}}$ is shown in figure \ref{fig:union}. Also, the more familiar concept of component can be defined as a graph that is not the union of any smaller parts and is also not part of a larger graph with the same property.
\item A graph is strongly connected if we can start at any node and get to any other node, respecting the directions of the arcs. $\mathcal{S}$, $\mathcal{S}_{\{4\}}$ and $\mathcal{S}_{\{1,2\}}$ are strongly connected, but $\mathcal{S}_{\{1,3\}}$ is not because there is no path from 3 to 1 in it.
\end{itemize}

\begin{figure}[hbt]
\centering
\includegraphics[width=0.5\linewidth]{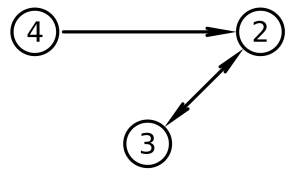}
\caption{The graph induced in $\mathcal{S}$ by the set $\{2,3,4\}$.}
\label{fig:induced}
\end{figure}

\begin{figure}[hbt]
\centering
\includegraphics[width=0.5\linewidth]{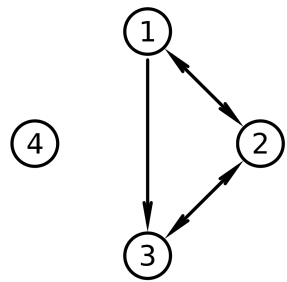}
\caption{An example of graph union. The graph $\mathcal{S}_{\{1,2,3\}}\cup\mathcal{S}_{\{4\}}$.}
\label{fig:union}
\end{figure}

Finaly, we will denote by $\mathcal{M}_{\Delta}$ the manifold with the states where only opinions in $\Delta$ survive:
\begin{equation}
\mathcal{M}_{\Delta} = \left\{\vec{\eta}\in Sim_M\,\,\left|\,\,\sum_{\sigma\in\Delta}\eta_{\sigma} = 1\right.\right\}.
\label{eq:M_Delta}
\end{equation}

\subsection{Mean-field results}

Let $\mathcal{R}$ be the skeleton of the confidence rule. The results for the mean-field fixed points are:

\begin{itemize}
\item Given a fixed point, its stability properties depend only on which opinions survive in it and on the skeleton of the confidence rule.

\item There exists a fixed point, where all opinions survive, iff $\mathcal{R}$ is a union of strongly connected graphs. Moreover, if $\mathcal{R}$ itself is strongly connected, this point is unique and an unstable node (the only exception is the case with 1 opinion, when the fixed point is the only point in the phase space). This is equivalent to the more intuitive statement that there exists a fixed point where all opinions coexist iff we cannot build a set $\Delta$ of opinions that can convince opinions in $\overline{\Delta}$, but cannot be convinced by any opinion in $\overline{\Delta}$.

\item The results concerning only opinions in a set $\Delta$ (the fixed points and the stabilities inside $\mathcal{M}_{\Delta}$) can be found using the model defined by the confidence rule $\mathcal{R}_{\Delta}$. (in other words, removing opinions from the model leaves us with a model with a different confidence rule, that is valid inside of $\mathcal{M}_{\Delta}$).

\item If $\mathcal{R}$ can be split in 2 independent models, that is, $\mathcal{R} = \mathcal{R}_{\Delta} \cup \mathcal{R}_{\overline{\Delta}}$. Then for all $\vec{\eta}\in \mathcal{M}_{\Delta}$ and $\vec{\nu}\in \mathcal{M}_{\overline{\Delta}}$, fixed points of the models with rules $\mathcal{R}_{\Delta}$ and $\mathcal{R}_{\overline{\Delta}}$ respectively, all the points $\alpha \vec{\eta} + (1-\alpha) \vec{\nu}$ with $0<\alpha<1$ are fixed points of the model with rule $\mathcal{R}$. Moreover, the number of unstable directions and stable directions is the respective sum of the numbers for $\vec{\eta}$ and $\vec{\nu}$ (when considering only directions inside $\mathcal{M}_{\Delta}$ and $\mathcal{M}_{\overline{\Delta}}$). The same thing is true for the neutral directions along which there is movement, but fails for the ones with no movement (we must add $\vec{\eta} - \vec{\nu}$ as an extra direction in this case).

\item A fixed point where only opinions in $\Delta$ survive is attractive iff $\Delta_{-} = \overline{\Delta}$. This also implies that $\Delta$ is a maximal independent set (see appendix \ref{ap:delta-maximal}) and hence that $\mathcal{M}_{\Delta}$ is an attractor.
\end{itemize}

These results have some interesting consequences and interpretations, that should be kept in mind when analysing the simulation results.

\begin{itemize}
\item The mean-field has no stable situations where 2 interacting opinions coexist. This means that all possible (static) attractors are of the form $\mathcal{M}_{\Delta}$, where $\Delta$ is a maximal independent set.
\item The requirement that $\Delta$ be maximal for $\mathcal{M}_{\Delta}$ to be an attractor allows the existence of attractors with surviving opinions that do not convince any opinions at all.
\item The condition $\Delta_{-} = \overline{\Delta}$ implies that it is possible to build confidence rules that have no such attractors. These rules display heteroclinic cycles (see appendix \ref{ap:heteroclinic}), which cause oscilations with diverging period and are heavily affected by finite size effects during simulations.
\item If every opinion can convince any other (that is $p_{\sigma\rightarrow\sigma'} \neq 0$ for all $\sigma\neq\sigma'$), then the consensus states are the only attractors.
\item In situations where part of the confidence rule can be broken in different components, by the removal of some opinions, there are manifolds where all points are fixed points, and these manifolds can be analised by putting together the analysis of each of the components.
\end{itemize}

\section{Simulation Results and Examples}
\label{sec:simulation}
For our simulations, we used Barab\'asi-Albert networks with $10^5$ sites and minimal connectivity equal to 5 (we used different networks, but always with these same parameters).

In order to compare trajectories obtained by simulations with trajectories obtained by integrating equations (\ref{eq:mean-field}) we recall that $\eta_{\sigma}$ is the expected value of the proportion of sites with opinion $\sigma$. Because of this and in order to reduce noise, we take averages over many simulations (one can also reduce noise by choosing a larger network size). More importantly, if the initial condition for the mean-field equations is $(\eta_1, \ldots, \eta_M)$, then this means that for the corresponding simulations, each site must have its opinion chosen at random with probability $\eta_{\sigma}$ for opinion $\sigma$.

The simulations we will do will be aimed at giving examples of the mean-field results from section \ref{ssec:confidence}, some of their counter-intuitive aspects and some divergences between the simulations and the mean-field.

\subsection{Attractors and stability}
\label{ssec:sim-attractors}

To illustrate the results about the stability properties of the fixed points, consider the rule $\mathcal{R}$, depicted in figure \ref{fig:attr} (actualy, a family of confidence rules). The maximal independent sets are $\Delta = \{1,2\}, \{1,5\}, \{3\}$ and $\{4\}$, but only $\{1,2\}$ and $\{1,5\}$ obey $\Delta_{-} = \overline{\Delta}$, meaning that the only stationary attractors are $\mathcal{M}_{\{1,2\}}$ and $\mathcal{M}_{\{1,5\}}$. After a transient we see one of 2 situations, the only surviving opinions will be 1 and 2 or they will be 1 and 5. We can see this from the time series of $\eta_{1} + \eta_{2} + \eta_{5}$ (it tends to 1) and $\eta_{2}.\eta_{5}$ (it tends to 0, although with a longer transient) (figures \ref{fig:attr-series-a} and \ref{fig:attr-series-b}).

\begin{figure}[hbt!]
\centering
\includegraphics[width=0.5\linewidth]{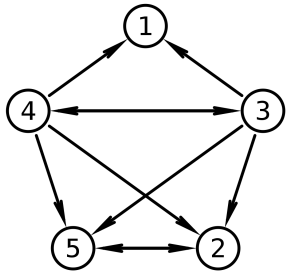}
\caption{The skeleton $\mathcal{R}$, of a confidence rule where 2 of the 4 maximal independent sets generate attractors. These static attractors can all be obtained by solving $\Delta_{-} = \overline{\Delta}$ in the rule, as pointed in section \ref{ssec:confidence}, and are $\mathcal{M}_{\{1,2\}}$ and $\mathcal{M}_{\{1,5\}}$.}
\label{fig:attr}
\end{figure}

\begin{figure}[htb!]
\centering
\subfigure[$\eta_{1} + \eta_{2} + \eta_{5}$] {\includegraphics[width=0.7\linewidth]{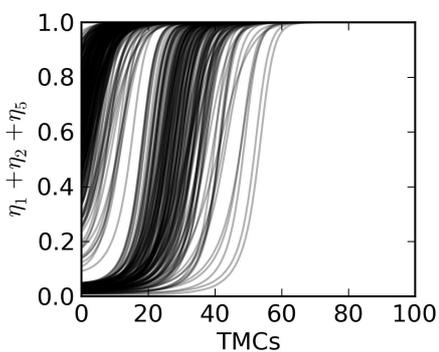}\label{fig:attr-series-a}} \\
\subfigure[$\eta_{2} . \eta_{5}$] {\includegraphics[width=0.7\linewidth]{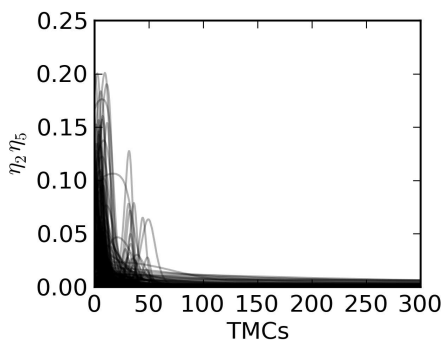}\label{fig:attr-series-b}}
\caption{Time series for the rule in figure \ref{fig:attr} with weights either 0 or 1, depicting the attractors. Note that in figure \ref{fig:attr-series-a} the ending value is 1 (meaning that opinions 3 and 4 do not survive). In figure \ref{fig:attr-series-b} the ending value is 0, showing that opinion 2 and 5 don't survive at the same time.}
\end{figure}

The other fixed points can be found by looking at the other induced graphs that are unions of strongly connected graphs. They are $\mathcal{R}_{\{3\}}$, $\mathcal{R}_{\{4\}}$, $\mathcal{R}_{\{3, 4\}}$ and $\mathcal{R}_{\{1,2,5\}}$. Note that $\mathcal{R}_{\{1,2,5\}} = \mathcal{R}_{\{1\}} \cup \mathcal{R}_{\{2,5\}}$ and that both components are strongly connected. This means that we will actualy have a line of fixed points connecting some point in the edge $\overline{P_2P_5}$ to the vertex $P_1$. The stability properties of these points are in table \ref{tab:fixed-points}. A projection of the phase space (where the weights in the confidence rule where taken as 0 or 1), showing the attractors and the features described in this table can be found in figure \ref{fig:attr-phase-space}.

\begin{table}[hbt!]
\centering
\begin{tabular}{|c|c|c|c|c|c|c|}
\hline
$\Delta$ & $\rule{0pt}{12pt}\overline{\Delta}$ & $\Delta_{-}$ & $\Delta_{+}$ & $u$ & $s$ & $n$ \\ \thickhline
3             & 1,2,4,5 & 4       & 1,2,4,5   & 3 & 1 & 0 \\ \hline
4             & 1,2,3,5 & 3       & 1,2,3,5   & 3 & 1 & 0 \\ \hline
3,4           & 1,2,5   & 3,4     & 1,2,3,4,5 & 4 & 0 & 0 \\ \hline
$1\times 2$,5 & 3,4     & 2,3,4,5 & 2,5       & 1 & 2 & 1 \\ \hline
\end{tabular}
\caption{The fixed points of the rule in figure \ref{fig:attr} that are not in attractors, denoted by the opinions that survive in them ($\Delta$). The line of fixed points connecting the edge $\overline{P_2P_5}$ to the vertex $P_1$ is denoted by 1$\times$2,5. For each fixed point, we list the number of unstable, stable and neutral directions ($u$, $s$ and $n$ respectively). The relation of these numbers with the sets $\overline{\Delta}, \Delta_{-}$, $\Delta_{+}$ and the number of components induced by $\Delta$ is given in section \ref{ssec:stability}.}
\label{tab:fixed-points}
\end{table}

\begin{figure}[hbt!]
\centering
\includegraphics[width=\linewidth]{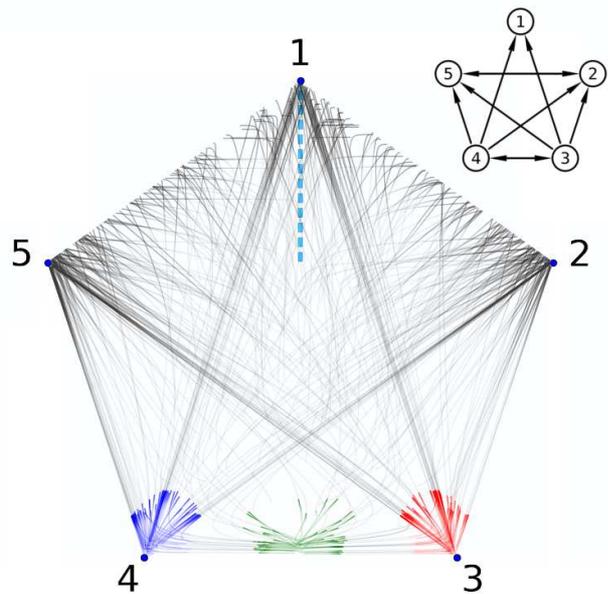}
\caption{(color online) Phase space projection,
depicting the structures described in table \ref{tab:fixed-points} and the attractors. On the top right we see a reordering of the skeleton of the rule making the independent sets more evident. The cyan dashed line shows the location of the saddle points ($1\times 2$,5 ; in table \ref{tab:fixed-points}), the blue shaded trajectories are passing near the point (4), the red ones near the point (3) and the green ones near the point (3,4). For all these trajectories and the gray ones, lighter shades indicate the beginning of the trajectories and darker shades indicate their ending. We can see then the trajectories going to the attractors $\mathcal{M}_{\{1,2\}}$ and $\mathcal{M}_{\{1,5\}}$, with some being at first attracted by the saddles in ($1\times 2$,5) before being repelled. We can also see the predicted stable direction for the fixed points (3) and (4) and the fact that (3,4) is an unstable node.}
\label{fig:attr-phase-space}
\end{figure}

\subsection{Surviving inert opinions}
\label{ssec:sim-inert}

Next, we consider two examples in which we have opinions that survive in an attractor, but don't convince any other opinion (we'll call them inert). Consider the rules $\mathcal{R}_1$ and $\mathcal{R}_2$ given in figure \ref{fig:iner}. In $\mathcal{R}_1$, $\mathcal{M}_{\{1,3\}}$ is an attractor, even though opinion 1 is inert (can't convince any of the others). In $\mathcal{R}_2$, $\mathcal{M}_{\{4,5\}}$, $\mathcal{M}_{\{1,2\}}$ and $\mathcal{M}_{\{2,3\}}$ are attractors, even though opinion 2 is inert.

\begin{figure}[hbt!]
\centering
\subfigure[$\mathcal{R}_1$]{
\includegraphics[width=0.4\linewidth]{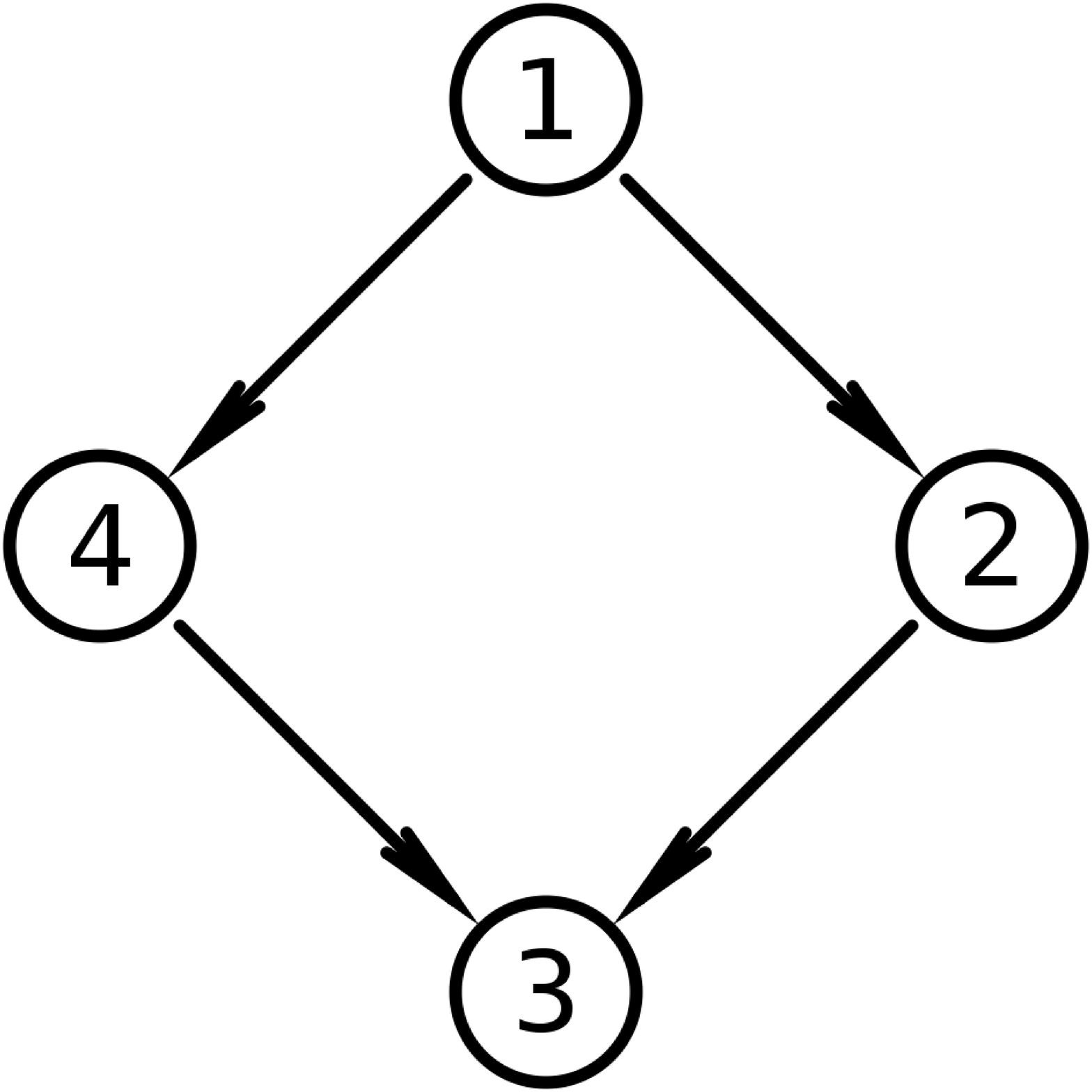}}
\subfigure[$\mathcal{R}_2$]{
\includegraphics[width=0.4\linewidth]{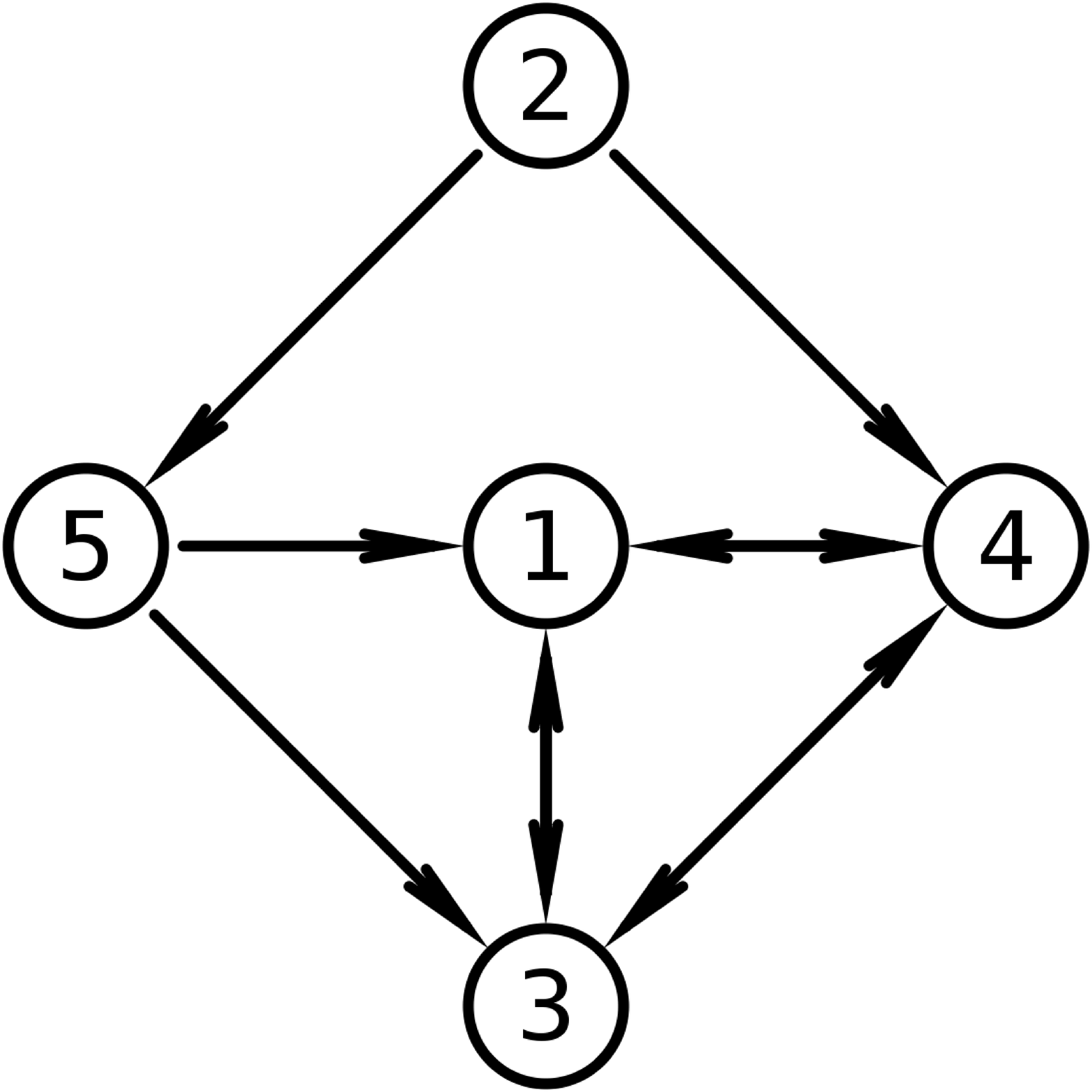}}
\caption{The skeleton of two confidence rules $\mathcal{R}_1$ and $\mathcal{R}_2$, such that inert opinions are able to survive in the stationary state.}
\label{fig:iner}
\end{figure}

We now check that this effect is present in the simulations. Time series for the models with confidence rules $\mathcal{R}_1$ and $\mathcal{R}_2$ (once again, the weights are taken as 0 or 1) are given in figures \ref{fig:quad_iner-att} to \ref{fig:pen_iner-iner}.

\begin{figure}[hbt!]
\centering
\subfigure[$(\mathcal{R}_1)$ $\eta_{1} + \eta_{3}$] {\includegraphics[width=0.7\linewidth]{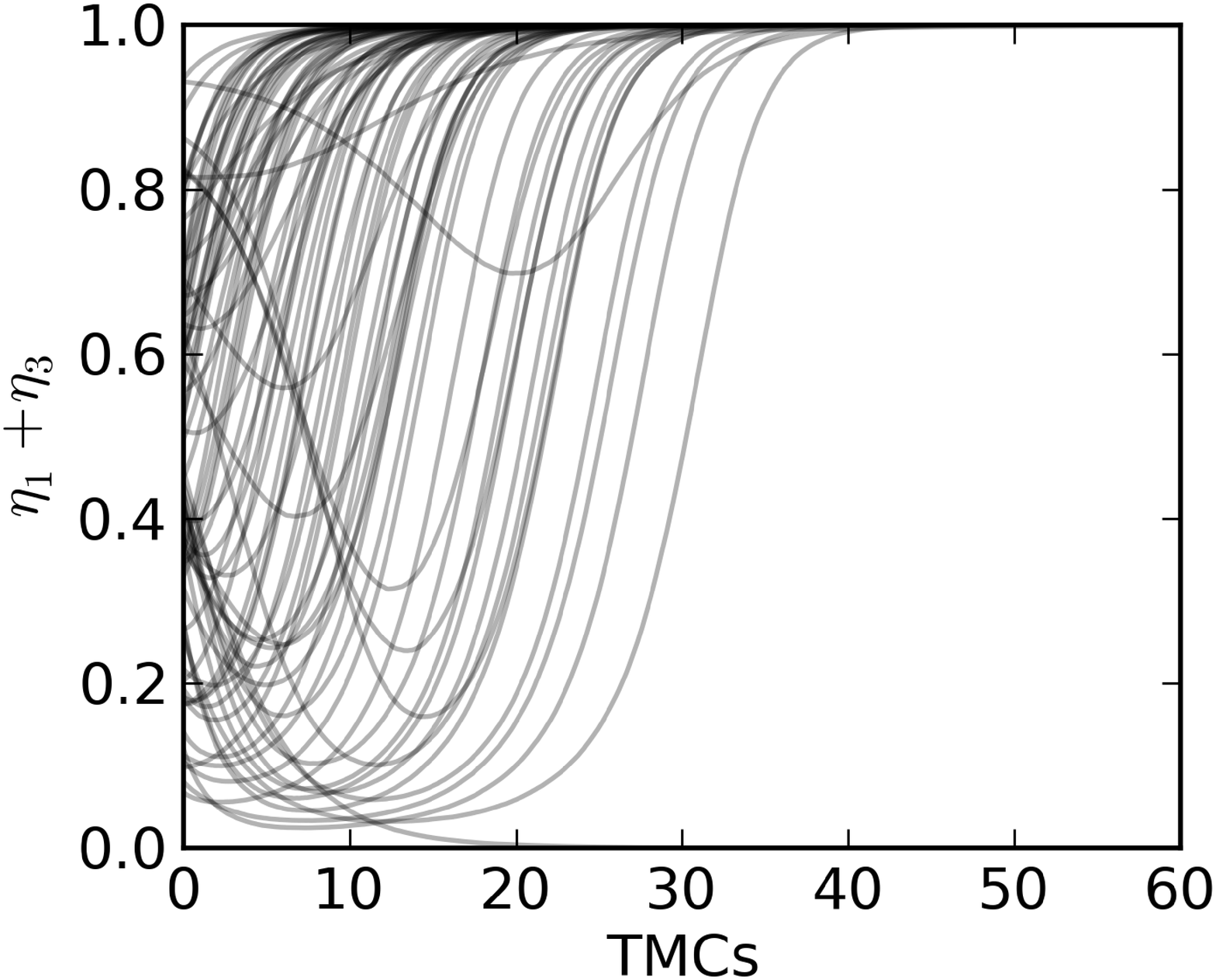}\label{fig:quad_iner-att}}
\captcont{Time series for the rules $\mathcal{R}_1$ and $\mathcal{R}_2$.}
\end{figure}

\begin{figure}[hbt!]
\centering
\subfigure[$(\mathcal{R}_1)$ $\eta_{1}$] {\includegraphics[width=0.7\linewidth]{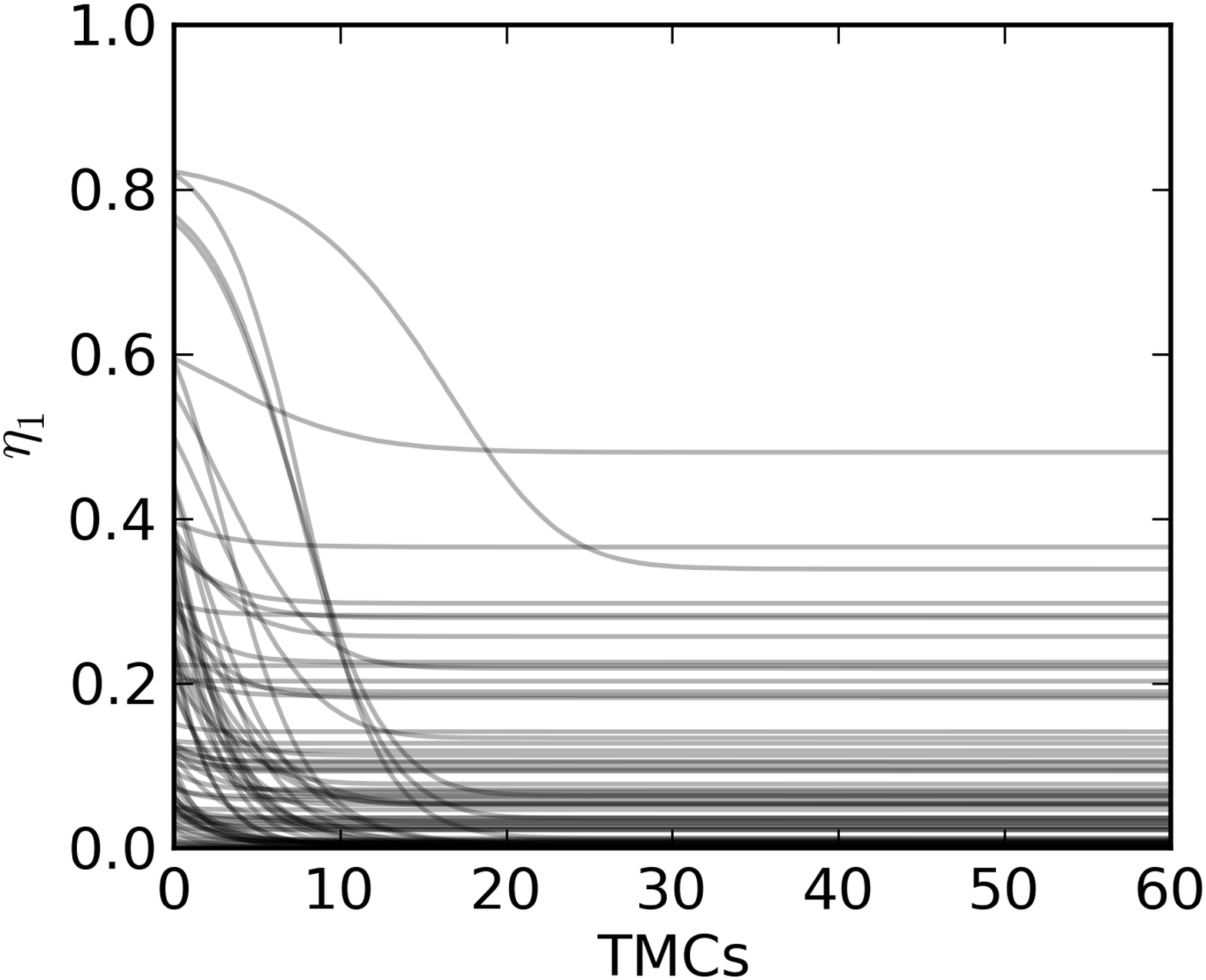}\label{fig:quad_iner-iner}}\\
\captcont*{(cont.)}
\end{figure}

\begin{figure}[hbt!]
\centering
\subfigure[$(\mathcal{R}_2)$ $\eta_{1} + \eta_2 + \eta_{3}$] {\includegraphics[width=0.7\linewidth]{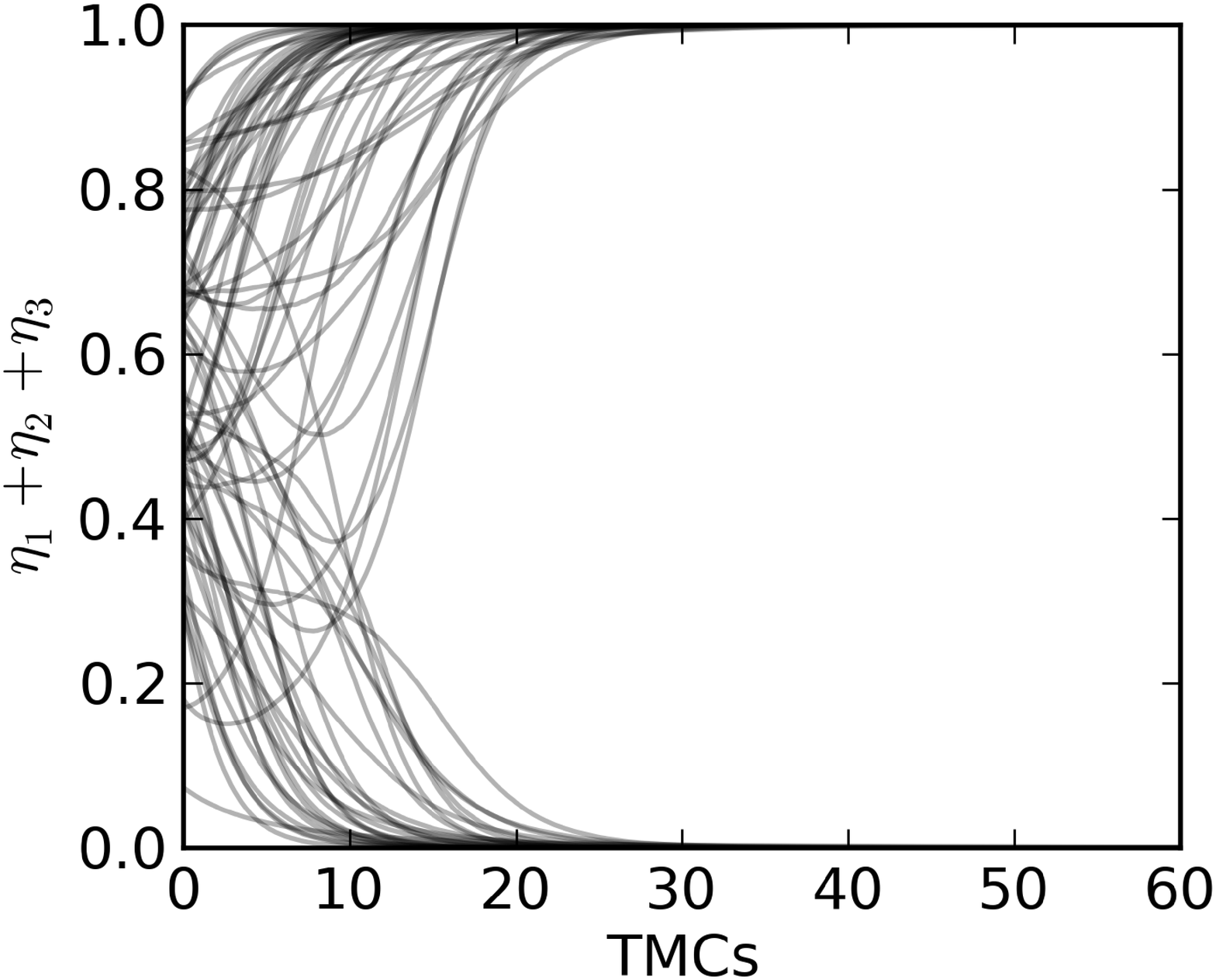}\label{fig:pen_iner-att}}\\
\captcont*{(cont.)}
\end{figure}

\begin{figure}[hbt!]
\centering
\subfigure[$(\mathcal{R}_2)$ $\eta_2$] {\includegraphics[width=0.7\linewidth]{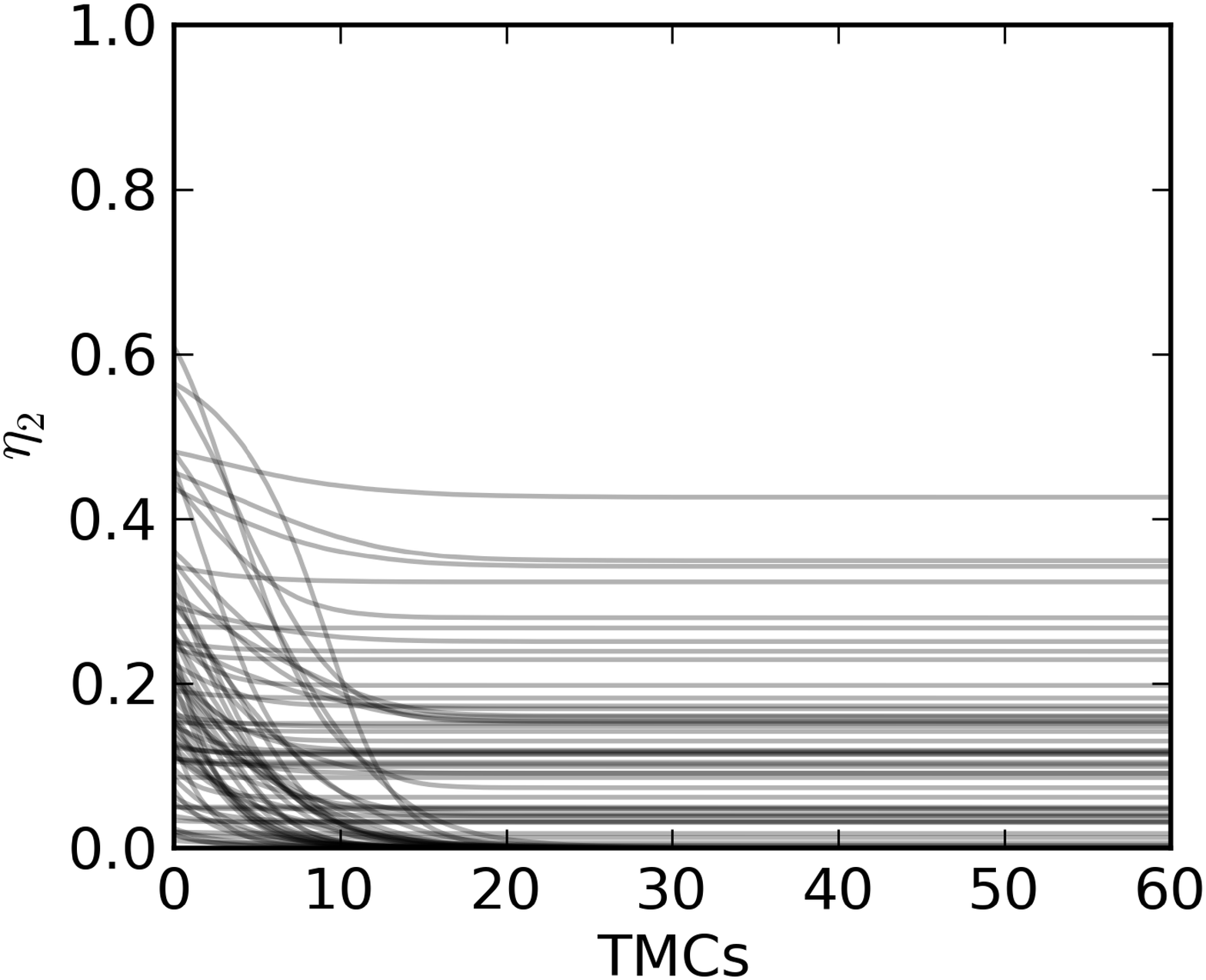}\label{fig:pen_iner-iner}}
\caption*{Time series for the rules $\mathcal{R}_1$ and $\mathcal{R}_2$ in figure \ref{fig:iner} with weights either 0 or 1. Graphs \ref{fig:quad_iner-att} and \ref{fig:pen_iner-att} depict time series containing the full attractor (the ending value is either 0 or 1 depending on the attractor reached), while graphs \ref{fig:quad_iner-iner} and \ref{fig:pen_iner-iner} focus in the surviving inert opinion (which always decays, but can reach a non-zero stationary value).}
\label{fig:series-iner}
\end{figure}

\subsection{Rules without stationary attractors}
\label{ssec:sim-no-attractors}

Consider a rule in which all opinions interact (that is, for all pair of distinct opinions $\sigma$ and $\sigma'$ either $p_{\sigma\rightarrow\sigma'} \neq 0$ or $p_{\sigma'\rightarrow\sigma} \neq 0$), but such that every opinion $\sigma$ has at least one different opinion $\sigma'$ that it can't convince. The independent sets of such rule are all unitary, but we have imposed that $\sigma'\notin\{\sigma\}_{-}$ and $\sigma\neq\sigma'$, so there are no solutions to $\Delta_{-} = \overline{\Delta}$ for this rule and hence it has no stationary attractors (it is possible to build other types of examples as well). An example for 4 opinions is given in figure \ref{fig:no-attr}.

\begin{figure}[hbt!]
\centering
\includegraphics[width=0.5\linewidth]{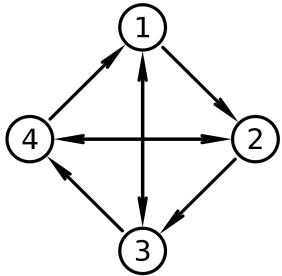}
\caption{A rule that has no attractors}
\label{fig:no-attr}
\end{figure}

In the appendix \ref{ap:heteroclinic} we prove that if a graph has no solutions to $\Delta_{-} = \overline{\Delta}$, then it has at least one directed cycle, where no edge is doubly connected. In the phase space, these cycles manifest themselves as heteroclinic cycles. These cycles will always be polygonal curves connecting the vertices of the simplex that correspond to the nodes the cycle in the graph goes through. Moreover, as the cycle goes through the nodes in $\Delta$, it means that $\Delta$ induces in the confidence rule a strongly connected graph with one component and as such, there exists an unstable fixed point where all the opinions in $\Delta$ coexist (in the example of figure \ref{fig:no-attr} the cycle would be $1\rightarrow 2\rightarrow 3\rightarrow 4\rightarrow 1$, corresponding to the heteroclinic cycle $\overline{P_1P_2P_3P_4P_1}$ and to an unstable fixed point where all opinions coexist). Hence, the heteroclinic cycle is fully contained in the border of $\mathcal{M}_{\Delta}$ and there is a fixed point in the bulk of $\mathcal{M}_{\Delta}$ that leads the trajectories to its border. The typical result is that as time goes by, the trajectories get closer to one of the cycles, which causes oscillations with a diverging period (as they pass each time closer to the consensus states, that are fixed points). In simulations, eventually a random fluctuation puts the system in a state where one of the opinions in the cycle gets extinct, leading the system to a stationary state.

\subsection{Long transients and stationary states}
\label{ssec:sim-transient}

In many simulations, there are situations in which the trajectories get stuck for long times in states that are not attractors. In some of these cases, the simulation got to a stationary state where there are no active connections between the agents (that is, a connection between a pair of agreeing sites and a neighbour that they can convince, according to the confidence rule). In other cases there are active connections, but some opinions appear in negligible amounts and the set $\Delta$ of opinions that are not negligible forms an independent set, but not a solution to $\Delta_{-} = \overline{\Delta}$. In the latter cases, the fixed points in $\mathcal{M}_{\Delta}$ are saddle points meaning that (this is shown in detail in section \ref{ssec:stability}) $\widetilde{\Delta}\equiv(\overline{\Delta} - \Delta_{-})\cap \Delta_{+} \neq \varnothing$ and usually, one (or more) of the negligible opinions will be able to rise again, causing long transients.

\begin{figure}
\includegraphics[width=0.5\linewidth]{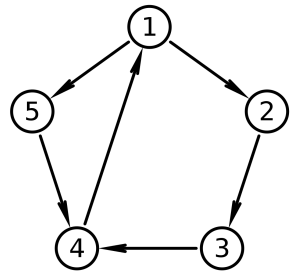}
\caption{A rule particulary prone to display long transients.}
\label{fig:pen_composite_rule}
\end{figure}

Considering the mean-field equations, if $\sigma\in\widetilde{\Delta}$ (meaning that $\sigma$ is negligible), then it evolves according to (see section \ref{ssec:stability} for further explanations)

\begin{equation}
\eta_{\sigma}(t) = \frac{\eta_{\sigma\oo}}{1 - \eta_{\sigma\oo} t \sum_{\sigma'\in\Delta}\eta_{\sigma'\oo} p_{\sigma'\rightarrow\sigma}},
\label{eq:long-transients-CM}
\end{equation}
as long as the opinions in $\overline{\Delta}$ are negligible. This implies that the time the trajectories spend close to these saddle points can be estimated, considering the time it takes for some of the opinions in $\widetilde{\Delta}$ to duplicate its proportion of sites in the network (all the other opinions in $\overline{\Delta}$ will remain negligible for much longer times, see appendix \ref{ap:last-order}). Solving \ref{eq:long-transients-CM} we get

\begin{equation}
\tau \simeq \min_{\sigma\in\widetilde{\Delta}} \left(\frac{1}{2\eta_{\sigma\oo}\sum_{\sigma'\in\Delta} \eta_{\sigma'\oo} p_{\sigma'\rightarrow\sigma}}\right).
\label{eq:return-time}
\end{equation}

We now verify this relation for the integration of the mean-field equations and compare these results with the simulations. We will use the rule in figure \ref{fig:pen_composite_rule}, with $\Delta = \{3,5\}$.

\begin{figure}[hbt!]
\centering
\subfigure[Mean Field] {\includegraphics[width=0.7\linewidth]{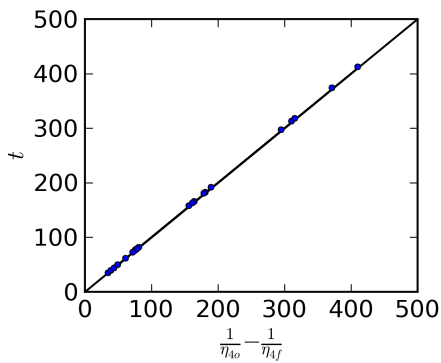}\label{fig:return-CM}}\\
\subfigure[Simulations in Barab\'asi-Albert networks] {\includegraphics[width=0.7\linewidth]{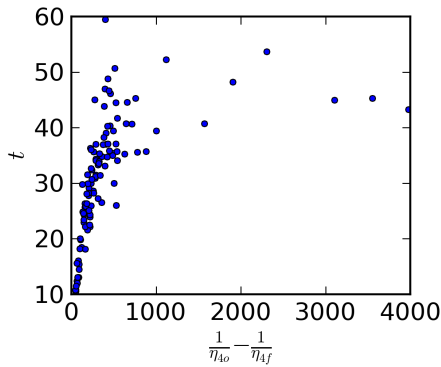}\label{fig:return-BA}}
\caption{(color online) Return times for $\Delta = \{3,5\}$ (more precisely the time the trajectory took since it crossed the surface $\eta_3 + \eta_5 > 1 - \varepsilon$, until it crossed the surface $\eta_4 > 2\varepsilon$, with $\varepsilon = 0.025$). We can see that simulations behave very differently than the mean-field. Particularly, return times are smaller than predicted by equation \ref{eq:return-pen} and the relationship between the two variables is not linear. The blue points correspond to the measured values and the black line corresponds to the prediction of equation \ref{eq:return-pen} (that holds only for the integration of the mean field equations).}
\end{figure}

For this choice of confidence rule and opinion set, we can approximate equation \ref{eq:long-transients-CM} with

\begin{equation}
\frac{1}{\eta_{4\oo}} - \frac{1}{\eta_4} \simeq t.
\label{eq:return-pen}
\end{equation}

The graphs for the mean field and the simulations can be found in figures \ref{fig:return-CM} and \ref{fig:return-BA} and show that if a simulation doesn't go to a stationary state, then it undergoes a transient much faster than what is expected from the mean-field equations.

\section{Analytical Results}
\label{sec:analytical}

\subsection{Stability and structure of the fixed points}
\label{ssec:stability}

Let us recall the mean field equation found in section \ref{ssec:confidence} (equation \ref{eq:mean-field}):

\[
\dot{\eta}_{\sigma} = r\sum_{\sigma'}\eta_{\sigma}\eta_{\sigma'}(\eta_{\sigma}^{q-1}p_{\sigma\rightarrow\sigma'} - \eta_{\sigma'}^{q-1}p_{\sigma'\rightarrow\sigma}).
\]
This equation is a flow where the phase space is the simplex defined by

\[
\sum_{\sigma} \eta_{\sigma} = 1\quad\mathrm{and}\quad \eta_{\sigma} \geq 0\,\forall\,\sigma
\]
and it is easy to show that all trajectories starting in the simplex never leave it (the sum $\sum_{\sigma} \dot{\eta}_{\sigma}$ is 0 because the term being summed in equation \ref{eq:mean-field} is skew-symmetrical and the sign of the variables doesn't change because $\nicefrac{\mathrm{d}\log\eta_{\sigma}}{\mathrm{dt}} = \nicefrac{\dot{\eta}_{\sigma}}{\eta_{\sigma}}$ is bounded in the bulk of the phase space).

Let us define $\vec{\eta} = (\eta_1, \eta_2, \eta_3, \ldots, \eta_M)$, $\dot{\eta}_{\sigma} = F_{\sigma}(\vec{\eta})$ and $\vec{F} = (F_1, F_2, F_3, \ldots, F_M)$, implying that $\dot{\vec{\eta}} = \vec{F}(\vec{\eta})$. Let us also denote the skeleton of the confidence rule by $\mathcal{R}$. The fixed points of equation \ref{eq:mean-field} are given by 

\begin{equation}
\left\{
\begin{array}{l}
\eta_{\sigma} = 0 \quad\mathrm{or}\\
\sum_{\sigma'} \eta_{\sigma}^{q-1}\eta_{\sigma'}p_{\sigma'\rightarrow\sigma} = \sum_{\sigma'} \eta_{\sigma'}^q p_{\sigma\rightarrow\sigma'}
\end{array}
\right.
\label{eq:fixed-point}
\end{equation}
for each opinion $\sigma$. If $\vec{\eta}\,^*$ is a fixed point (meaning, the represented point in phase space is a fixed point), then one can look at the behaviour of trajectories close to it, using the spectrum of the jacobian of $\vec{F}$, evaluated at $\vec{\eta}\,^*$. The jacobian of $\vec{F}$, $\mathcal{J}$ is given by

\[
\mathcal{J}_{\sigma,\sigma'} = 
r(\eta_{\sigma}^qp_{\sigma'\rightarrow\sigma} - q\eta_{\sigma}\eta_{\sigma'}^{q-1}p_{\sigma\rightarrow\sigma'}) +
\]
\begin{equation}
 + r\delta_{\sigma,\sigma'}\sum_{\sigma''}(q\eta_{\sigma}^{q-1}\eta_{\sigma''}p_{\sigma''\rightarrow\sigma} - \eta_{\sigma''}^qp_{\sigma\rightarrow\sigma''}).
\label{eq:jacobian}
\end{equation}

For each fixed point $\vec{\eta}\,^*$, we can define two sets of opinions, $\Delta$ and $\Omega \equiv \overline{\Delta}$, respectively the set of opinions that survive and that are extinct in $\vec{\eta}\,^*$. It is also useful to look at the components induced by $\Delta$: $\Delta_1, \ldots, \Delta_k$. That is,

\[
\mathcal{R}_{\Delta} = \bigcup_{i=1}^k \mathcal{R}_{\Delta_i}.
\]

The jacobian $\mathcal{J}$, can be simplified when it is evaluated at the fixed point and $\Omega \neq\varnothing$ (we will denote all quantities evaluated at the fixed point $\vec{\eta}\,^*$ by adding a $*$ superscript):

\begin{equation}
\mathcal{J}^*_{\sigma,\sigma'} = 
-r\delta_{\sigma,\sigma'}\sum_{\sigma''}\eta^{*q}_{\sigma''}p_{\sigma\rightarrow\sigma''},\,\,\mbox{if}\,\,\sigma\in\Omega
\label{eq:jacobian-omega}
\end{equation}
implying that the jacobian can be written up to permutations as 


\begin{equation}
\mathcal{J}^* = 
\begin{bmatrix}
\mathcal{J}^*_{\Delta} & \mathcal{N}\\
0 & \mathcal{J}^*_{\Omega}
\end{bmatrix},
\label{eq:jacobian-block-simple}
\end{equation}
where $\mathcal{J}^*_{\Delta}$ is the jacobian restricted to the opinions in $\Delta$ and $\mathcal{J}^*_{\Omega}$ is the one restricted to $\Omega$ (that is a diagonal matrix). It also follows from equation \ref{eq:jacobian} that $\mathcal{J}^*_{\Delta}$ can be written in a block-diagonal form where each block $\mathcal{J}_i^*$ is the jacobian restricted to the component $\Delta_i$:

\begin{equation}
\mathcal{J}^* = 
\begin{bmatrix}
\mathcal{J}_1^* & 0 & \ldots & 0 & \mathcal{N}_1\\
0 & \mathcal{J}_2^* & \ldots & 0 & \mathcal{N}_2\\
\vdots & \vdots & \ddots & \vdots & \vdots\\
0 & 0 & \ldots & \mathcal{J}_k^* & \mathcal{N}_k\\
0 & 0 & \ldots & 0 & \mathcal{J}^*_{\Omega}
\end{bmatrix},
\label{eq:jacobian-block-full}
\end{equation}
so in order to find the spectrum of $\mathcal{J}^*$, we need to put together the spectrum of $\mathcal{J}^*_{\Omega}$ (which is trivial to find) with the spectrum of each of the blocks $\mathcal{J}_i^*$. This means that all that is left is to find the spectrum of the jacobian in a coexistence fixed point, for a rule with one component. We show in the appendix \ref{ap:J-Delta} that, if $q\geq 2$, all the eigenvalues of $\mathcal{J}^*_{\Delta}$ have positive real part in this situation and the corresponding eigenvectors are parallel to the phase space, with the exception of one eigenvector that is not parallel and has eigenvalue 0 (with multiplicity 1). As this last eigenvector doesn't take points in the phase space to other points in the phase space it is irrelevant to the stability analysis. This null eigenvalue can be regarded as an artifact of having embedded an $(M-1)$-dimensional phase space in $M$ dimensions. However, according to our reasoning based in equation \ref{eq:jacobian-block-full}, it does mean that if $\Delta$ induces $k$ components, there are $k$ eigenvalues equal to 0. To understand what these null eigenvalues mean, consider the case $\Omega = \varnothing$. We define $\vec{\eta}_i$ and $\vec{F}_i$ as the vectors containing only the coordinates of $\vec{\eta}$ and $\vec{F}$ that are in $\Delta_i$, that is

\[
\vec{\eta} =
\begin{bmatrix}
\vec{\eta}_1 & \vec{\eta}_2 & \ldots & \vec{\eta}_k
\end{bmatrix}
\,\,\,\mbox{and}\,\,\,
\vec{F} =
\begin{bmatrix}
\vec{F}_1 & \vec{F}_2 & \ldots & \vec{F}_k
\end{bmatrix}.
\]
We also define the vectors $\vec{v}_i$ as 
\[
\vec{v}_1 =
\begin{bmatrix}
\vec{\eta}_1 & \vec{0} & \ldots & \vec{0}
\end{bmatrix}, 
\vec{v}_2 =
\begin{bmatrix}
\vec{0} & \vec{\eta}_2 & \vec{0} & \ldots & \vec{0}
\end{bmatrix}, 
\]\[
, \ldots,
\vec{v}_k =
\begin{bmatrix}
\vec{0} & \ldots & \vec{\eta}_k 
\end{bmatrix}.
\]

For every $i$, $\vec{F}_i$ is a function of only $\vec{\eta}_i$. Moreover, $\vec{F}_i$ is homogeneous with order $q+1$, so if

\begin{equation}
\eta_{\Delta_i} = \sum_{\sigma\in\Delta_i} \eta_{\sigma}\,\,\,\mbox{and}\,\,\,\vec{\zeta}_i = \frac{\vec{v}_i^*}{\eta_{\Delta_i}^*} \Rightarrow \vec{F}_i(\alpha\vec{\zeta}_i) = \vec{0}.
\end{equation}
It follows that all the points given by

\begin{equation}
\sum_{i=1}^k\alpha_i \vec{\zeta}_i\,\,\,\mbox{such that}\,\,\,\sum_{i=1}^k \alpha_i = 1\,\,\,\mbox{and}\,\,\,\alpha_i \geq 0\,\,\forall\,\, i
\label{eq:degenerate-fixed-points}
\end{equation}
are fixed points of the model. The set of points defined by equation \ref{eq:degenerate-fixed-points} is the convex hull of the points defined by the $\vec{\zeta}_i$ (we identify these vectors with points in the phase space using equation \ref{eq:embedding}) and as these vectors are linearly independent, it means the convex hull must have $k-1$ dimensions.

Euler's theorem gives us another consequence of the homogeneity of $\vec{F}_i$, namely
\begin{equation}
\mathcal{J}_i \vec{\eta}_i = (q+1)\vec{F}_i.
\label{eq:theorem-euler}
\end{equation}
So in the fixed point we have $\mathcal{J}_i^* \vec{\eta}_i^* = \vec{0} \Rightarrow \mathcal{J}^* \vec{v}_i^* = \vec{0}$, yielding a base of eigenvectors with eigenvalue 0. These eigenvectors can be reorganized so that they are still linearly independent, but only one of them points outwards from the phase space, meaning that said eigenvector is irrelevant to the stability analysis (the embedding artifact previously discussed). Such a base can be taken as $\vec{\zeta}_i - \vec{\eta}\,^*$, $i = 1, \ldots, k-1$, together with $\vec{\eta}\,^*$. Removing the artifact, we have a base that generates the flat defined by equation \ref{eq:degenerate-fixed-points}. As all the points in the flat are fixed points, this means that there is no movement along these directions for trajectories close to the fixed points.

Finaly, we need to study the spectrum of $\mathcal{J}^*_{\Omega}$. As this is a diagonal matrix, we can obtain the eigenvalues directly:

\begin{equation}
\lambda_{\sigma} = -r\sum_{\sigma'\in\Delta}\eta^{*q}_{\sigma'}p_{\sigma\rightarrow\sigma'} \leq 0,
\label{eq:lambda-sigma}
\end{equation}
for each $\sigma\in\Omega$ and the eigenvalue $\lambda_{\sigma}$ tells us if the trajectories are attracted or repelled to the manifold $\eta_{\sigma}=0$.

It follows that if $\mathcal{R}$ is such that $\Omega \subseteq \Delta_{-}$, then $\lambda_{\sigma} < 0$ for all $\sigma\in\Omega$ and we have a complete picture of the behaviour of trajectories close to the fixed point. On the other hand, if $\Omega \nsubseteq \Delta_{-}$, some of these eigenvalues are null and we need to examine higher orders. In order to do that, if $\lambda_{\sigma} = 0$, then we must make a Taylor expansion of $F_{\sigma}$ around $\vec{\eta}\,^*$. As $F_{\sigma}$ is given by

\[
r\sum_{\sigma'}\eta_{\sigma}^q\eta_{\sigma'}p_{\sigma'\rightarrow\sigma} - r\sum_{\sigma'\in\Omega}\eta_{\sigma}\eta_{\sigma'}^{q}p_{\sigma\rightarrow\sigma'},
\]
then it is easy to see that the lower order term, different from 0 yields 

\begin{equation}
\dot{\eta}_{\sigma} \simeq r\eta_{\sigma}^q\sum_{\sigma'\in\Delta}\eta_{\sigma'}^*p_{\sigma'\rightarrow\sigma},
\label{eq:high-order}
\end{equation}
meaning that the trajectories are repelled from the manifold $\eta_{\sigma} = 0$, unless $\sigma\notin\Delta_{+}$. In particular, for $q=2$ and $r=1$, the solution of equation \ref{eq:high-order} reads

\[
\eta_{\sigma}(t) = \frac{\eta_{\sigma\oo}}{1 - \eta_{\sigma\oo} t \sum_{\sigma'}\eta_{\sigma'}^* p_{\sigma'\rightarrow\sigma}},
\]
as stated in section \ref{ssec:sim-transient} (equation \ref{eq:long-transients-CM}).

Let $\mathcal{M}_{\Delta}$ be the manifold where only opinions in $\Delta$ survive. We then know that if all $\lambda_{\sigma} < 0$, the trajectories get attracted to $\mathcal{M}_{\Delta}$, but if $\lambda_{\sigma} = 0$ and $\sigma\in\Delta_{+}$ for any $\sigma\in\Omega$, then the trajectories get repelled from $\mathcal{M}_{\Delta}$. Finaly, suppose that either $\lambda_{\sigma} < 0$ or $\sigma\notin\Delta_{+}$ for all opinions in $\Omega$. Let $\omega$ be the set of opinions such that $\lambda_{\sigma} = 0$,

\[
\Lambda \equiv \max_{\sigma\in\Omega-\omega} \lambda_{\sigma} < 0\,\,\mbox{and}\,\,\eta_{\omega} \equiv \sum_{\sigma\in\omega}\eta_{\sigma}.
\]
We show in the appendix \ref{ap:last-order} that starting in a sufficiently close neighbourhood of $\eta_{\omega}=0$ the following inequality holds

\begin{equation}
\eta_{\omega\oo} e^{\nicefrac{\vert\Omega-\omega\vert r}{q\Lambda}} \leq \eta_{\omega} \leq \eta_{\omega\oo}e^{-\nicefrac{\vert\Omega-\omega\vert r}{\Lambda}},
\label{eq:assymptotic-degeneracy}
\end{equation}
and so trajectories are neither attracted to nor repelled from $\mathcal{M}_{\Delta}$.

We can now put all these results together. If $\vec{\eta}\,^*$ is a fixed point, such that $\mathcal{R}_{\Delta}$ has $k$ components, then the trajectories in a neighbourhood of $\vec{\eta}\,^*$ are such that (remembering that $\Omega = \overline{\Delta}$):
\begin{itemize}
\item There are $|\Delta| + |(\overline{\Delta} - \Delta_{-})\cap\Delta_{+}| - k$ unstable directions.
\item There are $|\overline{\Delta}\cap\Delta_{-}|$ stable directions.
\item There are $k - 1$ directions along which there is no movement.
\item There are $|(\overline{\Delta} - \Delta_{-}) \cap (\overline{\Delta} - \Delta_{+})|$ directions that are neither attractive nor repulsive, but along which there is movement.
\end{itemize}

It follows from the equations \ref{eq:mean-field} that if $\Delta$ is independent, then $\mathcal{M}_{\Delta}$ is composed entirely of fixed points. On the other hand, if $\vec{\eta}\,^*$ is attractive, then it has no unstable directions and the only neutral directions are those along which there is no movement, so if $k$ is the number of components of $\mathcal{R}_{\Delta}$, then $|\Delta|\geq k$ and it follows that

\[
|\Delta| + |(\overline{\Delta} - \Delta_{-})\cap\Delta_{+}| - k = 0 \Rightarrow |(\overline{\Delta} - \Delta_{-})\cap\Delta_{+}| = 
\]
\begin{equation}
= k - |\Delta| \leq 0 \Rightarrow |(\overline{\Delta} - \Delta_{-})\cap\Delta_{+}| = 0 \Rightarrow |\Delta| = k.
\label{eq:attractor}
\end{equation}
$|\Delta| = k$ implies that $\Delta$ is independent. The condition that there are no neutral directions yields

\begin{equation}
(\overline{\Delta} - \Delta_{-})\cap\Delta_{+} = \varnothing\,\,\mbox{and}\,\,(\overline{\Delta} - \Delta_{-}) \cap (\overline{\Delta} - \Delta_{+}) = \varnothing.
\label{eq:attractor2}
\end{equation}
This implies that

\[
((\overline{\Delta} - \Delta_{-})\cap\Delta_{+})\cup((\overline{\Delta} - \Delta_{-}) \cap (\overline{\Delta} - \Delta_{+})) = \varnothing \Leftrightarrow
\]\[
(\overline{\Delta} - \Delta_{-})\cap(\Delta_{+}\cup(\overline{\Delta} - \Delta_{+})) = \varnothing.
\]
As $\Delta$ is independent, $\Delta_{+} \subseteq \overline{\Delta}$ and so $\Delta_{+}\cup(\overline{\Delta} - \Delta_{+}) = \overline{\Delta}$, implying

\[
(\overline{\Delta} - \Delta_{-})\cap\overline{\Delta} = \varnothing \Leftrightarrow \overline{\Delta} - \Delta_{-} = \varnothing \Leftrightarrow \overline{\Delta} \subseteq \Delta_{-}.
\]
As $\Delta$ is independent we have $\Delta_{-} \subseteq \overline{\Delta} \Rightarrow \overline{\Delta} = \Delta_{-}$. So, as $\overline{\Delta} = \Delta_{-}$ always solves equations \ref{eq:attractor}, it follows that $\mathcal{M}_{\Delta}$ is attractive iff $\Delta$ is independent and $\overline{\Delta} = \Delta_{-}$. As we show in appendix \ref{ap:delta-maximal}, $\overline{\Delta} = \Delta_{-}$ alone implies that $\Delta$ is maximal independent, so to find the attractors of the model it suffices to find all the sets $\Delta$, satisfying $\Delta_{-} = \overline{\Delta}$ and the corresponding manifolds $\mathcal{M}_{\Delta}$.

\subsection{Existence of the fixed points}
\label{ssec:existence}

In the last section, we analysed the stability properties of a fixed point, supposing that it existed (given a fixed point, what its stability properties are). Now we check when a fixed point, where only the opinions in a given set $\Delta$ coexist, exists. Our analysis of the case in which the set $\Delta$ of surviving opinions induces more than one component in the confidence rule shows us that we only need to study the case in which all the opinions coexist and the confidence rule has one component.

In this case, any fixed point where all opinions coexist must be an unstable node, if it exists, and so if we had embedded our phase space in $M-1$ instead of $M$ dimensions (substituting $\eta_M$ by $1-\sum_{\sigma \neq M}\eta_{\sigma}$, for example), the jacobian $\widetilde{\mathcal{J}}$ of the corresponding flux would be a real matrix that is positive definite when evaluated in such a fixed point, implying that $\det (\widetilde{\mathcal{J}}^*) > 0$. It follows that the index of the fixed point is 1 \footnote{The index in the case when the jacobian is not singular equals the sign of its determinant. More information about indices and their meaning can be found in most textbooks about differential geometry.} and that we can apply the implicit function theorem in this case.

In the appendix \ref{ap:TPH}, we use these informations together with the Poincar\'e-Hopf theorem to show that if our confidence rule is such that the directed skeleton is a complete directed graph (a graph where there is a doubly connected edge between any two nodes) then there exists exactly one fixed point where all opinions coexist. On the other hand, a confidence rule $\mathcal{R}$ can be regarded as a point in the parameter space and for every confidence rule, there exists arbitrarily small neighbourhoods of it in this space, containing rules whose directed skeletons are complete. So for every confidence rule $\mathcal{R}$, there exists a path in the parameter space leading to it, but such that all other rules in the path have complete skeleton. Finaly, if our rule has a complete skeleton, then we can apply the implicit function theorem for its coexistence fixed point, meaning that the fixed point changes continuously for continuous changes in the parameters (changes in the rule). More importantly, the jacobian evaluated in the fixed point changes continuously, as well as its eigenvalues.

Suppose then that we have a rule $\mathcal{R}$ in which there are no coexistence fixed points. We build a path ending in $\mathcal{R}$, where all rules have complete skeleton and we look at the coexistence fixed point along the path. The limit of the fixed point as the rule tends to $\mathcal{R}$ must also be a fixed point (as $\vec{F}$ is changing continuously too). Examining the eigenvalues of the jacobian evaluated at the fixed point along the path, with the exception of the limit fixed point, the real parts of the eigenvalues must all be positive. For the limit fixed point, there exists a set $\Omega$ of non-surviving opinions and $\mathcal{J}_{\Omega}$ must have real non-positive eigenvalues, implying that $\mathcal{J}_{\Omega} = 0$ (by continuity). Recalling equation \ref{eq:lambda-sigma}, this implies that if $\Delta$ is the set of surviving opinions in the limit fixed point, then $\Delta_{-}\cap\overline{\Delta} = \varnothing$.

On the other hand, suppose that we have a rule with one component and a set $\Delta \neq \varnothing, V$ ($V$ is the set of all nodes) such that $\Delta_{-}\cap\overline{\Delta} = \varnothing$. As we have only one component, we must have $\Delta_{+}\cap\overline{\Delta} \neq \varnothing$, and so if we define $\eta_{\Delta} = \sum_{\sigma\in\Delta}\eta_{\sigma}$, then

\[
\dot{\eta}_{\Delta} = -r\sum_{\sigma\in\Delta}\sum_{\sigma'\in\overline{\Delta}\cap\Delta_{-}} \eta_{\sigma}\eta^q_{\sigma'}p_{\sigma\rightarrow\sigma'},
\]
meaning $\dot{\eta}_{\Delta} < 0$ in the whole region of the phase space where all opinions coexist. This implies that there exists a fixed point where all opinions coexist iff there is no set of opinions $\Delta \neq \varnothing, V$, obeying $\Delta_{-}\cap\overline{\Delta} = \varnothing$. We show in the appendix \ref{ap:sources-connectivity}, that this is equivalent to saying the graph is strongly connected. So going back to our previous results about the structure of the fixed points, we have that there exists a fixed point where only the opinions in $\Delta$ survive iff $\Delta$ induces an union of strongly connected graphs (as such a fixed point exists iff it exists for each of the components separately).

A similar argument can be used to prove uniqueness. Suppose that $\mathcal{R}$ is a rule with one component that has coexistence fixed points. We can build a path ending in $\mathcal{R}$ going only through rules with complete skeleton. But the coexistence fixed points are unique along this path and as we can apply the implicit function theorem, any coexistence fixed point of $\mathcal{R}$ must be a limit fixed point of the rules in the path, implying that it is also unique. In the case where $\Delta$ induces $k$ strongly connected graphs, this implies that there exists one and only one $(k-1)$-dimensional flat where only opinions in $\Delta$ survive.

\section{Conclusions}
On this work, we expanded our previous results about the Sznajd model with general confidence rules (interpreted here as biases and prejudices), giving analytical results about the existence, structure and stability properties of the fixed points in the mean-field case, finding a very rich behaviour. We gave simulation results in Barab\'asi-Albert networks that show examples of this mean-field behaviour and showed some of the discrepances between the model simulated in these networks and the integration of the mean-field equations.

Even though neither the equations for the fixed points can be solved analiticaly, nor can the exact eigenvalues of the Jacobian be all determined, our dynamical systems approach was still able to determine the sign of the real parts of these eigenvalues and the higher order behaviours, when these were needed. Surprisingly, this analysis showed us that the various properties of the fixed points depend only on a few qualitative properties of the confidence rule (the directed skeleton). This, in turn, allowed us to make a connection between the mean-field results and graph theory problems and this connection can even be used to study more complex behaviours, like the heteroclinic cycles in the phase space that always appear in the absence of attractors.

In regard to the simulations, most of the discrepances with the mean-field seem to come from the existence of frozen states that don't correspond to mean-field attractors, but that can be reached by the model in a network. It is not entirely clear if these are purely finite size effects, but their origin suggests that they should be more common as the number of opinions increases and that the introduction of a random noise, in which opinions change randomly with a given probability, should destroy this effect. A curious finding in the confidence rule studied in section \ref{ssec:sim-transient} is that when simulations got close to the frozen states, but managed to get away from them, they took much less time than would be expected from the mean-field results (it must be stressed that we only investigated this behaviour for this confidence rule).

Given the simple conclusions that were reached and the generality of our model (we would like to stress that the mean-field results are valid not only for the Sznajd model but for the $q$-voter model with $q\geq 2$), we believe that similar approaches might be fruitful in other models where asymmetrical interactions exist, way beyond opinion propagation and sociophysics, like infection spreading and ecology models. It would also be interesting to see if similar connections with graph theory problems exist in other models and, if they do, how rich they are.

\begin{center}
{\bf Acknowledgements}
\end{center}
AMT would like to thank C. C. Gomes for discussions that were crucial to the results in the appendix \ref{ap:TPH}. Both authors thank FAPESP for financial support.

\appendix

\section{Graph theory theorems with applications to our model}
\subsection{$\overline{\Delta} = \Delta_{-}$ implies maximal independence}
\label{ap:delta-maximal}

\begin{theorem}
Let $G$ be a graph and let $\Delta$ be a set of nodes in it such that $\overline{\Delta} = \Delta_{-}$. This implies that $\Delta$ is a maximal independent set.
\end{theorem}

\begin{proof}
To see this, suppose that $\Delta_{-} = \overline{\Delta}$ but $\Delta$ is not independent, then it follows that there exists $\sigma, \sigma'\in \Delta$ such that $\sigma\in \{\sigma'\}_{+} \Leftrightarrow \sigma'\in \{\sigma\}_{-} \Rightarrow \sigma'\in\Delta \cap \Delta_{-} \neq \varnothing \Rightarrow \Delta \cap \overline{\Delta} \neq \varnothing$, which is a contradiction.

So if $\Delta_{-} = \overline{\Delta}$ then $\Delta$ is independent. If $\overline{\Delta} = \varnothing$ then it is trivial that $\Delta$ is maximal. If $\overline{\Delta} \neq \varnothing$, take $\sigma \in \overline{\Delta}$. It follows that

\[
(\Delta\cup\{\sigma\})_{-} = \Delta_{-}\cup\{\sigma\}_{-} = \overline{\Delta}\cup\{\sigma\}_{-} \Rightarrow 
\]\[
\Rightarrow (\Delta\cup\{\sigma\})\cap (\Delta\cup\{\sigma\})_{-} = 
(\Delta\cup\{\sigma\})\cap (\overline{\Delta}\cup\{\sigma\}_{-}) = 
\]\[
= ((\Delta\cup\{\sigma\})\cap \overline{\Delta})\cup ((\Delta\cup\{\sigma\})\cap \{\sigma\}_{-}) \supseteq
\]\[
\supseteq (\Delta\cup\{\sigma\})\cap \overline{\Delta} = \{\sigma\} \neq \varnothing \Rightarrow (\Delta\cup\{\sigma\})\cap (\Delta\cup\{\sigma\})_{-}\neq \varnothing,
\]
which implies that $\Delta\cup\{\sigma\}$ is not independent and hence, $\Delta$ is maximal.
\end{proof}

\subsection{Relation between the absence of attractors and heteroclinic cycles}
\label{ap:heteroclinic}

\begin{theorem}
Let $G$ be a directed graph such that no set of nodes obeys $\overline{\Delta} = \Delta_{-}$, then there exists a directed cycle in $G$ that doesn't use any of the doubly linked edges.
\end{theorem}

\begin{proof}
Suppose that there is no such cycle in $G$ and let $G'$ be the graph $G$ after removing all the doubly linked edges. By hypothesis, $G'$ is a directed acyclic graph and so a topological ordering in $G'$ is possible. This means that we can define a strict partial order in $V(G)$:
\[
i\prec j \mbox{ iff there is a path from } j \mbox{ to } i \mbox{ in } G'.
\]
We can also restrict this order to a subset $\Omega$ of $V(G)$, such that $i,j\in\Omega \Rightarrow i \prec_{\Omega} j$ iff $i\prec j$ (note that this is not the same thing as saying the path exists in $G'_{\Omega}$). Consider now the set $\Delta$ built from the following algorithm:
\begin{enumerate}
\item Attribute $\Delta \leftarrow \varnothing$, $\Xi \leftarrow \varnothing$ and $\Omega \leftarrow V(G)$.
\item If $\Omega$ equals $\varnothing$ stop, else let $i$ be a minimal element of $\prec_{\Omega}$.
\item Remove $i$ from $\Omega$ and add it to the set $\Delta$.
\item Remove the elements from the predecessor of $i$ with respect to $G$ that are in $\Omega$, ($i_{-}(G)\cap\Omega$), from $\Omega$ and add them to $\Xi$.
\item go to 2.
\end{enumerate}
By construction the set $\Xi$ obeys $\overline{\Delta} = \Xi \subseteq \Delta_{-}$ (the predecessor with respect to $G$). Moreover, the set $\Delta$ is independent. To see this, suppose that at some time during the construction of $\Delta$, there are no connections in $G$ between nodes in $\Delta$ and nodes in $\Omega$ when we reach step 2 (this is trivially true for the starting iteration) and let $i$ be the minimal element of $\prec_{\Omega}$ chosen in this step. As $i$ is minimal, there are no nodes in $\Omega$ such that $j\prec_{\Omega} i$ and hence there are no paths from $i$ to any other element in $\Omega$ in the graph $G'$ and hence $i_{+}(G)\cap\Omega$ contains only nodes that are connected to $i$ through doubly connected edges, implying $i_{+}(G)\cap\Omega \subseteq i_{-}(G)\cap\Omega$ and so in step 4 we are transfering all the nodes in $\Omega$, that had any connection with $\Delta$ after step 3, to the set $\Xi$. So after an iteration of the algorithm there are still no connections between nodes in $\Delta$ and $\Omega$ when reach step 2 again (and so by induction, this holds during the whole contruction of $\Delta$). But as the nodes are added to $\Delta$ from $\Omega$ one at a time, adding a new node won't add connections between nodes in $\Delta$, implying that $\Delta$ remains independent during its whole construction. On the other hand, this implies $\Delta_{-}\subseteq\overline{\Delta}$. Recalling that by the construction of $\Xi$ we have $\overline{\Delta} \subseteq \Delta_{-}$, it follows that $\Delta_{-}=\overline{\Delta}$.
\end{proof}

The relevance of this theorem to our problem is that a solution to $\overline{\Delta}=\Delta_{-}$ in the skeleton of the rule is equivalent to a static attractor in the phase space and saying the cycle $\sigma_1\rightarrow\sigma_2\rightarrow\ldots\rightarrow\sigma_1$ in this skeleton has no doubly connected edges is equivalent to saying that the polygonal curve $\overline{P_{\sigma_1}P_{\sigma_2}\ldots P_{\sigma_1}}$ is a heteroclinic cycle, meaning that every rule that has no static attractors must have at least one heteroclinic cycle.

\subsection{Necessary and sufficient condition for a graph to be strongly connected}
\label{ap:sources-connectivity}

Let $G$ be a graph and $\Delta\neq\varnothing$ a set nodes.

\begin{definition}
$\Delta$ is a sink (source) iff it obeys $\Delta_{+}\cap\overline{\Delta} = \varnothing$ ($\Delta_{-}\cap\overline{\Delta} = \varnothing$). In both cases, $\Delta$ is called minimal if there is no non-empty proper subset of it with the same property.
\end{definition}

\begin{definition}
The span of a node $i$, $i_{span}$ is the set of all nodes $j$ in $G$, such that $j$ can be reached from $i$. The span of a set of nodes $\Delta$ is defined as the union of the span of each of its nodes.
\end{definition}

\begin{corollary}
Every span is a sink and if $\Delta$ is a sink, then $\Delta = \Delta_{span}$.
\end{corollary}

\begin{corollary}
As no arc leaves a sink, if $X$ is a sink and $Y\subseteq X$ then $Y_{span} = Y_{span}(G_X)$, the span of $Y$ in $G_X$.
\end{corollary}

\begin{theorem}
A sink (source) is minimal iff it induces a strongly connected graph.
\end{theorem}

\begin{proof}
Let $\Delta$ be a sink that induces a strongly connected graph in $G$. Suppose by absurd that $\Delta$ in not minimal, then there exists $\Gamma \subset \Delta$, such that $\Gamma$ is also a sink and $\Gamma\neq\varnothing$. Let $\omega \in \Delta - \Gamma$. As $G_{\Delta}$ is strongly connected, then for all $i,j\in\Delta$ we have $i\in j_{span}(G_{\Delta})$ and hence $\omega \in \Gamma_{span}(G_{\Delta})$. As $\Gamma \subset \Delta$ and $\Delta$ is a sink in $G$, it follows that $\Gamma_{span}(G_{\Delta}) = \Gamma_{span}(G)$ and as $\Gamma$ is also a sink in $G$ we have $\Gamma_{span}(G) = \Gamma$. But this implies $\omega \in \Gamma$, which is a contradiction and so $\Delta$ must be minimal.

On the other hand, if $\Delta$ is a minimal sink in $G$ and we suppose by absurd that $G_{\Delta}$ is not strongly connected, there exists $i,j\in\Delta$ such that $i\notin j_{span}(G_{\Delta})$. $\Delta$ is a sink and $\{j\}\subseteq\Delta$, so this means $i\notin j_{span}(G)\subseteq$ and $j_{span}(G)\subseteq \Delta_{span}(G) = \Delta$. But then $j_{span}(G)\subseteq \Delta - \{i\}$ and so $j_{span}(G)$ is a non-empty proper subset of $\Delta$ that is a sink, contradicting the assumption that $\Delta$ was minimal. Hence, $G_{\Delta}$ must be strongly connected.

The proof for sources is obtained considering the graph $G'$, obtained by switching the orientation of all the arcs of $G$ (which transforms sinks in sources and vice-versa, but keeps the same induced graphs strongly connected)
\end{proof}

The relevance of this to our problem is that when the confidence rule has only one component, the condition that we found for the existence of a coexistence fixed point can be rephrased as saying that the set of all nodes is a minimal source. This theorem shows then that this is equivalent to saying the confidence rule is strongly connected, which makes more easy to see what the result for many components is.

\section{Topology and matrix theory theorems}
\label{ap:theorems}

\begin{theorem}[Poincar\'e-Hopf]
Let $\mathcal{M}$ be a compact, orientable and differentiable manifold and let $\vec{F}$ be a vector field defined in $\mathcal{M}$, such that it has only isolated zeros (every zero has an open neighbourhood in which it is unique). If either $\mathcal{M}$ has no border or if $\vec{F}$ points outwards (acording to the orientation of $\mathcal{M}$) along all points of the border, then the sum of the indices of all the zeros of $\vec{F}$ in the interior of $\mathcal{M}$ equals the Euler characteristic of $\mathcal{M}$.
\end{theorem}

\begin{theorem}[Gershgorin]
Let $M\in \mathbb{M}_n(\mathbb{C})$ be a square matrix whose general term is $m_{i,j}$. So if $\lambda$ is an eigenvalue of $M$, then there exists an $i$ such that
\[
|\lambda - m_{i,i}| \leq \sum_{j\neq i} |m_{i,j}|.
\]
\end{theorem}

\begin{theorem}[Levy-Desplanques]
Let $M\in \mathbb{M}_n(\mathbb{C})$ be an irreducible square matrix whose general term is $m_{i,j}$. If
\[
|m_{i,i}| \geq \sum_{j\neq i} |m_{i,j}|\,\forall\, i
\]
and there exists an $i$ such that
\[
|m_{i,i}| > \sum_{j\neq i} |m_{i,j}|,
\]
then $\det (M) \neq 0$.
\end{theorem}

\begin{theorem}
Let $M\in \mathbb{M}_n(\mathbb{C})$ be a symmetrical irreducible square matrix whose general term is $m_{i,j}$. If $m_{i,i} \neq 0$ for some $i$, then for all $k$ such that $1\leq k < n$, there exists an irreducible principal submatrix of $M$ with order $k$.
\end{theorem}

The next theorem is a strengthening of a theorem found in \cite{Johnson-1974}.

\begin{theorem}
Let $M\in \mathbb{M}_n(\mathbb{C})$ be a square matrix and let its hermitian part be $H = \nicefrac{(M + M^{\dagger})}{2}$. If $H$ is positive semidefinite, with the multiplicity of 0 equal to $\mu$, then for all $T\in \mathbb{M}_n(\mathbb{C})$, such that $T$ is hermitian positive definite, then $MT$ (and $TM$) is positive semidefinite and the sum of the geometric multiplicities of its eigenvalues with null real part is smaller or equal than $\mu$.
\end{theorem}

\begin{proof}
By our hypothesis, the eigenvalues of $H$ are non-negative real numbers and its eigenvectors can be arranged as an orthonormal basis. We can split this basis in 2 parts, $\{u_i\}$, with the eigenvectors with eigenvalue 0 and $\{v_i\}$, for the others. Define $\lambda_i > 0$, the eigenvalue such that $Hv_i = \lambda_iv_i$ and define $U$, the linear span of $\{u_i\}$. Let $x$ be a column vector and $x^{\dag}$ its conjugate transpose, so
\[
x' = \sum_i \alpha_iu_i,\quad x'' = \sum_j \beta_jv_j,\quad\mathrm{and}\quad x=x'+x'' \Rightarrow
\]\[
x^{\dag} Hx = (x'^{\dag} +x''^{\dag})H(x'+x'') =
\]\[
=(Hx')^{\dag}(x' + x'') + x''^{\dag} (Hx') + x''^{\dag} Hx'' = x''^{\dag} Hx'' = 
\]\[
= \sum_{i,j} \beta_i^* \beta_j v^{\dag}_i Hv_j = \sum_{i,j} \beta_i^* \beta_j \lambda_j \delta_{i,j} = \sum_{i} \left|\beta_i\right|^2\lambda_i.
\]
\noindent Hence, $x^{\dag} Hx > 0\Leftrightarrow x \notin U$ . On the other hand $2\mathrm{Re}(x^{\dag} Mx) = x^{\dag} Mx + (x^{\dag} Mx)^* = x^{\dag}(M + M^{\dag})x = 2x^{\dag} Hx$. Let $S$ be a nonsingular matrix and $w$, a normalized eigenvector of $S^{\dag} MS$, with eigenvalue $\gamma$. Taking $x = Sw$, it follows
\[
\mathrm{Re}(\gamma) = \mathrm{Re}(\gamma w^{\dag} w) =
\]\[
= \mathrm{Re}(w^{\dag} S^{\dag} MSw) = \mathrm{Re}(x^{\dag} Mx) = x^{\dag}Hx.
\]
\noindent Define $W = S^{-1}U = \{y\in W \Leftrightarrow Sy \in U\}$, so $\mathrm{Re}(\gamma)>0 \Leftrightarrow w\notin W$. As the dimension of $U$ is $\mu$, it follows that $W$ also has dimension $\mu$. The sum $\sigma$ of the geometric multiplicities of the eigenvalues of $S^{\dag} MS$ with null real part is the dimension of the linear span, $N$, of the corresponding eigenvectors. As all these eigenvectors belong to $W$ and $W$ is a linear subspace, then it follows that $N$ is a subspace of $W$ and hence $\sigma \leq \mu$.

All the properties of the spectrum of a matrix (including algebraic and geometric multiplicities) are encoded in its Jordan canonical form, and this form is invariant by similarity transformations, so $S^{\dag} MS$, $MSS^{\dag}$ and $SS^{\dag}M$ have the same spectral properties. This proves the theorem, as any hermitian positive definite matrix $T$ can be written as $SS^{\dag}$, with a non-singular $S$ using a Cholesky decomposition.
\end{proof}

\section{Spectrum of the jacobian for a coexistence point in a rule with only one component}
\label{ap:J-Delta}

Let $\vec{\eta}\,^*$ be a coexistence fixed point (that is, all opinions survive) in a model with a rule that has only one component and at least 2 opinions. We recall that $\mathcal{J}^*\vec{\eta}\,^* = \vec{0}$ (due to homogeneity) and that if $\vec{1} = (1,\ldots, 1)$, then $\vec{1}\mathcal{J}^* = \vec{0}$ (this follows from the conservation of the sum of the variables). Let $D$ be the diagonal matrix whose diagonal terms are the coordinates of $\vec{\eta}\,^*$ (in other words $D = \mathrm{diag}(\vec{\eta}\,^*)$), then the symmetric matrix $A$, defined as

\begin{equation}
A = \mathcal{J}^*D + (\mathcal{J}^*D)^T
\label{eq:def-A}
\end{equation}
has off-diagonal terms given by

\begin{equation}
A_{\sigma,\sigma'} = r(1-q)(\eta_{\sigma}^{*q}\eta_{\sigma'}^*p_{\sigma'\rightarrow\sigma} + \eta_{\sigma}^{*}\eta_{\sigma'}^{*q}p_{\sigma\rightarrow\sigma'}) \leq 0
\label{eq:A-term}
\end{equation}
and each of the rows (columns) of $A$ sum 0. Moreover, as the confidence rule has only one component, $A$ is irreducible, implying that at least one off-diagonal term in each row is different from 0 and hence all the diagonal terms are positive. Finaly, we can use these informations to apply Gershgorin's theorem (appendix \ref{ap:theorems}) and find that $A$ is positive semidefinite.

Denote the principal submatrix of $A$, obtained by removing row and column $\sigma$ by $A^{(\sigma)}$. If $A_{\sigma,\sigma'} \neq 0$ is an off-diagonal term from $A$, both $A^{(\sigma)}$ and $A^{(\sigma')}$ are such that one of the rows has a positive sum. As every row has an off-diagonal term different from 0, then all principal submatrices of $A$ with order $M-1$ have at least one row that has a positive sum. Applying Gershgorin's theorem again, we find that all the $A^{(\sigma)}$ are positive semidefinite. Moreover, as $A$ is irreducible, then there exists $\sigma$ such that $A^{(\sigma)}$ is also irreducible and by the Levy-Desplanques theorem (appendix \ref{ap:theorems}), $A^{(\sigma)}$ must be positive definite. Putting everything together, we find that

\[
\det(A) = 0 \,\,\mbox{and}\,\,\sum_{\sigma}\det(A^{(\sigma)}) > 0,
\]
meaning that 0 is an eigenvalue of $A$ with algebraic multiplicity 1. As $A$ is also hermitian, then all the other eigenvalues must be real and positive.

As $A^{(\sigma)}$ can be written as $A^{(\sigma)} = \mathcal{J}^{*(\sigma)}D^{(\sigma)} + (\mathcal{J}^{*(\sigma)}D^{(\sigma)})^T$ (because $D$ is diagonal), we can apply our theorem from appendix \ref{ap:theorems} to $A^{(\sigma)}$ to find that all the $\mathcal{J}^{*(\sigma)}$ are positive semidefinite and at least one $\mathcal{J}^{*(\sigma)}$ is positive definite. As these matrices are real, this implies that $\sum_{\sigma} \det(\mathcal{J}^{*(\sigma)}) > 0$. Finaly, we can apply this same theorem to $A$ to find that $\mathcal{J}^*$ is positive semidefinite and the sum of the geometric multiplicities of all eigenvalues with null real part is at most one, implying there is at most one eigenvalue with null real part. As $\mathcal{J}^{*}.\vec{\eta}\,^* = \vec{0}$, it follows that 0 is the only eigenvalue of $\mathcal{J}^{*}$ with null real part. Moreover, as $\sum_{\sigma} \det(\mathcal{J}^{*(\sigma)}) > 0$ it has algebraic multiplicity 1. As we also have $\vec{1}\mathcal{J}^* = \vec{0}$, then $\vec{1}$ is a left eigenvector with eigenvalue 0, implying that if $\vec{v}$ is a right eigenvector with eigenvalue different from $0$, then $\vec{1}.\vec{v} = 0$.

Putting these results together, all the eigenvalues have positive real part and the corresponding eigenvectors are parallel to the phase space, with the exception of the eigenvector $\vec{\eta}\,^*$, that is not parallel and has eigenvalue 0 (with multiplicity 1).

\section{High order stability analysis for fixed points in which opinions get extinct}
\label{ap:last-order}

Suppose that we have a fixed point in which only opinions in $\Delta$ survive and let $\Omega = \overline{\Delta}$. For each $\sigma\in\Omega$, we define $\lambda_{\sigma}$ as

\[
\lambda_{\sigma} = -r\sum_{\sigma'\in\Delta}\eta^{*q}_{\sigma'}p_{\sigma\rightarrow\sigma'} \leq 0.
\]

Suppose that either $\lambda_{\sigma} < 0$ or $\sigma\notin\Delta_{+}$ for all opinions in $\Omega$ and let $\omega$ be the set of opinions such that $\lambda_{\sigma} = 0$,

\[
\Lambda \equiv \max_{\sigma\in\Omega-\omega} \lambda_{\sigma} < 0\,\,\mbox{and}\,\,\eta_{\omega} \equiv \sum_{\sigma\in\omega}\eta_{\sigma}.
\]
It follows from the first order analysis we did in section \ref{ssec:stability}, that if $\sigma\in\Omega-\omega$ and the initial value of $\eta_{\sigma}$, $\eta_{\sigma \oo}$ is sufficiently close to 0, then $\eta_{\sigma}$ evolves as

\[
\eta_{\sigma}(t) = \eta_{\sigma \oo}e^{\lambda_{\sigma}t},
\]
as long as all opinions in $\Omega$ remain negligible. It follows from the mean field equations that

\begin{equation}
\dot{\eta}_{\omega} = r\sum_{\sigma\in\omega}\sum_{\sigma'\in\Omega-\omega} \eta_{\sigma}\eta_{\sigma'}(\eta_{\sigma}^{q-1}p_{\sigma'\rightarrow\sigma} - \eta_{\sigma'}^{q-1}p_{\sigma\rightarrow\sigma'}).
\label{eq:eta-omega}
\end{equation}
So if $\eta_{\omega}$ is sufficiently close to 0, it evolves as

\[
\eta_{\sigma} \simeq \eta_{\sigma \oo}e^{\lambda_{\sigma} t}\,\,\forall\,\,\sigma\in\Omega-\omega\Rightarrow
\]\[
\dot{\eta}_{\omega} \simeq r\sum_{\sigma\in\omega}\sum_{\sigma'\in\Omega-\omega} (\eta_{\sigma' \oo}e^{\lambda_{\sigma'} t}\eta_{\sigma}^q p_{\sigma'\rightarrow\sigma} - \eta_{\sigma' \oo}^qe^{q\lambda_{\sigma'} t}\eta_{\sigma} p_{\sigma\rightarrow\sigma'})
\]
yielding the following inequalities

\[
-r\sum_{\sigma\in\omega}\sum_{\sigma'\in\Omega-\omega}\eta_{\sigma' \oo}^qe^{q\lambda_{\sigma'} t}\eta_{\sigma}p_{\sigma\rightarrow\sigma'} \leq \dot{\eta}_{\omega} \leq
\]\[
\leq r\sum_{\sigma\in\omega}\sum_{\sigma'\in\Omega-\omega}\eta_{\sigma' \oo}e^{\lambda_{\sigma'} t}\eta_{\sigma}^q p_{\sigma'\rightarrow\sigma} \Rightarrow
\]\[
-r\sum_{\sigma\in\omega}\sum_{\sigma'\in\Omega-\omega}e^{q\lambda_{\sigma'} t}\eta_{\sigma} \leq \dot{\eta}_{\omega} \leq r\sum_{\sigma\in\omega}\sum_{\sigma'\in\Omega-\omega}e^{\lambda_{\sigma'} t}\eta_{\sigma} \Rightarrow
\]\[
-r\sum_{\sigma\in\omega}\sum_{\sigma'\in\Omega-\omega}e^{q\Lambda t}\eta_{\sigma} \leq \dot{\eta}_{\omega} \leq r\sum_{\sigma\in\omega}\sum_{\sigma'\in\Omega-\omega}e^{\Lambda t}\eta_{\sigma} \Rightarrow
\]\[
-r\vert\Omega-\omega\vert e^{q\Lambda t}\eta_{\omega} \leq \dot{\eta}_{\omega} \leq r\vert\Omega-\omega\vert e^{\Lambda t}\eta_{\omega} \Rightarrow
\]\[
-r\vert\Omega-\omega\vert e^{q\Lambda t} \leq \frac{\mathrm{d}}{\mathrm{dt}}\ln({\eta}_{\omega}) \leq r\vert\Omega-\omega\vert e^{\Lambda t}.
\]
Integrating in time and taking the limit $t\rightarrow\infty$ gives the inequalities \ref{eq:assymptotic-degeneracy}:

\begin{equation}
\eta_{\omega\oo} e^{\nicefrac{\vert\Omega-\omega\vert r}{q\Lambda}} \leq \eta_{\omega} \leq \eta_{\omega\oo}e^{-\nicefrac{\vert\Omega-\omega\vert r}{\Lambda}},
\end{equation}
and so trajectories are neither attracted to nor repelled from $\mathcal{M}_{\Delta}$. This ensures that the whole reasoning is consistent, as it is always possible to make $\eta_{\omega\oo}$ sufficiently small, so that the hypothesis of small $\eta_{\omega}$ always holds.

\section{Applying the Poincar\'e-Hopf theorem to the case of a confidence rule with complete directed skeleton}
\label{ap:TPH}

Consider a rule with a skeleton corresponding to a complete directed graph and define

\[
p = \min_{\sigma\neq\sigma'} \{p_{\sigma\rightarrow\sigma'}\}.
\]

In order to apply the Poincar\'e-Hopf theorem, we build a family of manifolds that includes the phase space:

\[
M_{\epsilon} = \left\{\vec{\eta}\in Sim_M \Big| \eta_{\sigma} \geq \epsilon \,\,\forall\,\,\sigma\right\}.
\]
These manifolds all satisfy the hypothesis of the theorem and the borders of $M_{\epsilon}$ are given by the facets $\eta_{\sigma} = \epsilon$ (that is, we are using $M$ dimensions to define our manifolds, but we are embedding them in $M-1$ dimensions). The fixed points we obtain for the flow in the mean field equation are not isolated when we look at the problem in $M$ dimensions (because of the homogeneity of the equations), but our results about the Jacobian show that embedding the phase space in $M-1$ dimensions instead of $M$ is enough to isolate the zeros (this follows from applying the implicit function theorem. Another way of isolating the zeros would be to add a term in the equation that is 0 inside the phase space, but is different from 0 outside, but this has the downside of making the hypothesis to be checked more complicated). It also follows from the spectrum of this jacobian that the indices of any fixed points in the interior of any of the manifolds $M_{\epsilon}$ would be 1.

The last hypothesis to be checked is then that the vector field $\vec{F}$ points outside along the border. In the manifold $M_{\epsilon}$, this is equivalent to the following statement:

\begin{equation}
\eta_{\sigma} = \epsilon \Rightarrow F_{\sigma} = r\epsilon^q\sum_{\sigma'}\eta_{\sigma'}p_{\sigma'\rightarrow\sigma} - r\epsilon\sum_{\sigma'}\eta_{\sigma'}^qp_{\sigma\rightarrow\sigma'} < 0.\Rightarrow
\label{eq:tph-condition}
\end{equation}

\[
\epsilon^{q-1}\sum_{\sigma'}\eta_{\sigma'}p_{\sigma'\rightarrow\sigma} < \sum_{\sigma'}\eta_{\sigma'}^qp_{\sigma\rightarrow\sigma'}.
\]
As the left hand side is smaller than $\epsilon^{q-1}$ and the right hand side is greater than $\nicefrac{p}{M^{q-1}}$, then it suffices to take

\begin{equation}
\epsilon < \frac{p^{\nicefrac{1}{(q-1)}}}{M}
\label{eq:tph-condition2}
\end{equation}
in order to get a manifold $M_{\epsilon}$ such that we can apply the theorem.

The Euler characteristic of all of the $M_{\epsilon}$ is 1, meaning that if we can apply the theorem there exists exactly one fixed point in its interior. Together with equation \ref{eq:tph-condition2}, this means that there exists exactly one coexistence fixed point and it obeys

\[
\eta_{\sigma} \geq \frac{p^{\nicefrac{1}{(q-1)}}}{M}\,\,\forall\,\,\sigma.
\]

\bibliography{andre}
\bibliographystyle{plain}
\end{document}